\documentclass[letterpaper,11pt]{article}
\usepackage{geometry}
\usepackage{amsmath,color}
\usepackage{commath}
\usepackage{amsthm}
\usepackage{amssymb}
\usepackage{setspace,mathrsfs}
\usepackage{algorithm,natbib}
\usepackage{algorithmic}
\usepackage{url,amsmath,amsfonts,amssymb}
\usepackage{caption,subcaption,algorithm,graphicx}
\usepackage{listings,enumerate,color,relsize,multirow,rotating}
\usepackage{booktabs}

\def\bzero{\mathbf{0}}

\newtheorem{prop}{Proposition}
\newtheorem{assume}{Assumption}
\newtheorem{condition}{Condition}

\newtheorem{theorem}{Theorem}
\newtheorem{remark}{Remark}

\newcommand{\trans}{^{\intercal}}

\def\Tsc{\mathcal{T}}

\def\bgamma{\boldsymbol{\gamma}}
\def\bbeta{\boldsymbol{\beta}}

\def\supone{^{\scriptscriptstyle (1)}}

\def\Isc{\mathcal{I}}

\def\bgammahat{\widehat{\bgamma}}

\def\U{\boldsymbol{U}}
\def\Esc{\mathcal{E}}
\def\Fsc{\mathcal{F}}

\def\Nsc{\mathcal{N}}

\def\ep{\ep }

\def\X{\boldsymbol{X}}

\def\x{\boldsymbol{x}}
\def\balpha{\boldsymbol{\alpha}}

\def\bmu{\boldsymbol{\mu}}
\def\supmfk{^{[\text{-}k]}}
\def\supmfkj{^{[\text{-}k,\text{-}j]}}

\def\submfkj{_{\text{-}k,\text{-}j}}
\def\bC{\boldsymbol{C}}
\def\D{\boldsymbol{D}}
\def\Lscr{\mathscr{L}}
\def\Pbb{\mathbb{P}}
\definecolor{darkred}{RGB}{150,50,50}
\definecolor{brown}{RGB}{250,100,100}
\definecolor{green}{RGB}{000,150,100}
\definecolor{purple}{RGB}{250,000,180}

\def\ep{\mathbb{E}}
\def\bfeta{\boldsymbol{\eta}}
\def\subhd{_{\sf \scriptscriptstyle HD}}
\def\subml{_{\sf \scriptscriptstyle ML}}
\def\Bsc{\mathcal{B}}
\def\Esc{\mathcal{E}}
\def\supone{^{\scriptscriptstyle (1)}}
\def\suptwo{^{\scriptscriptstyle (2)}}
\def\Z{\boldsymbol{Z}}

\begin{document}

\title{Double/Debiased Machine Learning for Logistic Partially \\ Linear Model}
\author{Molei Liu$^{1}$, Yi Zhang$^{2}$, Doudou Zhou$^{3}$}

\footnotetext[0]{This paper is a complete and more advanced version of a recent technical note by the author. See the note from \url{https://arxiv.org/abs/2008.12467}.}
\footnotetext[1]{Department of Biostatistics, Harvard Chan School of Public Health.}
\footnotetext[2]{Department of Statistics, Harvard University}
\footnotetext[3]{Department of Statistics, University of California, Davis.}

\maketitle

\begin{abstract}
\noindent 

We propose double/debiased machine learning approaches to infer (at the parametric rate $n^{-1/2}$) the parametric component of a logistic partially linear model with the binary response following a conditional logistic model of a low dimensional linear parametric function of some key (exposure) covariates and a nonparametric function adjusting for the confounding effect of other covariates. We consider a Neyman orthogonal (doubly robust) score equation consisting of two nuisance functions: nonparametric component in the logistic model and conditional mean of the exposure on the other covariates and with the response fixed. To estimate the nuisance models, we separately consider the use of high dimensional (HD) sparse parametric models and more general (typically nonparametric) machine learning (ML) methods. In the HD case, we derive certain moment equations to calibrate the first order bias of the nuisance models and grant our method a model double robustness property in the sense that our estimator achieves the desirable $O(n^{-1/2})$ rate when at least one of the nuisance models is correctly specified and both of them are ultra-sparse. 
In the ML case, the non-linearity of the logit link makes it substantially harder than the partially linear setting to use an arbitrary conditional mean learning algorithm to estimate nuisance component of the logistic model. We handle this obstacle through a novel ``full model refitting" procedure that is easy-to-implement and facilitates the use of nonparametric ML algorithms in our framework. Our ML estimator is rate doubly robust in the same sense as \cite{chernozhukov2016double}. We evaluate our methods through simulation studies and apply them in assessing the effect of emergency contraceptive (EC) pill on early gestation foetal with a policy reform in Chile in 2008 \citep{bentancor2017assessing}.

\end{abstract}

\section{Introduction}

Consider a logistic partially linear model. Let $\{(Y_i,A_i,\X_i):i=1,2,\ldots,n\}$ be independent and identically distributed samples of $Y\in\{0,1\}$, $A\in\mathbb{R}$ and $\X\in\mathbb{R}^p$. Assume that
\begin{equation}
\Pbb (Y=1\mid A,\X)={\rm expit}\{\beta_0 A+r_0(\X)\},
\label{model}
\end{equation}
where ${\rm expit}(\cdot)={\rm logit}^{-1}(\cdot)$, ${\rm logit}(a)=\log\{a/(1-a)\}$ and $r_0(\cdot)$ is an unknown nuisance function of $\X$. In an experimental or observational study with $A$ taken as the exposure variable, $Y$ being the binary response  of interests and $\X$ representing the  observed confounding variables, parameter $\beta_0$ is of particular interests as it measures the conditional effect of $A$ on $Y$ in the scale of logarithmic odds ratio. And as the most common and natural way to characterize the conditional model of a binary outcome against some exposure, model (\ref{model}) has been extensively used in economics and policy science studies. 

Our goal is to estimate and infer $\beta_0$ asymptotic normally at the rate $n^{-1/2}$. When $\X$ is a scalar and $r_0(\cdot)$ is smooth, classic semiparametric kernel or sieve regression works well for this purpose \citep{severini1994quasi,lin2006semiparametric}. When $\X$ is of high dimensionality, these approaches can have poor performance due to curse of dimensionality and one would rather estimate $r_0(\cdot)$ with modern HD  (parametric) or ML (nonparametric)\footnote{Our HD setting refers to HD parametric (linear or generalized linear) model and the ML setting refers to ML models of conditional mean estimation (prediction/classification) that is blackbox and usually nonparametric.} methods that are much more resistant to the growing dimensionality and complexity of $\X$. However, unlike the partially linear model scenario \citep{chernozhukov2018double,dukes2020inference}, robust and efficient inference of $\beta_0$ in (\ref{model}) with HD or ML nuisance models is yet to be well-studied. 


In recent, \cite{tan2019doubly} proposed a simple and flexible doubly robust estimator to enhance the robustness to the potential misspecification of $r(\x)$ specified as a fixed-dimensional parametric function: $r(\x)=\x\trans\bgamma$. They introduced a parametric model $m(\x)=g(\x\trans\balpha)$ with a (known) link function $g(\cdot)$ for the conditional mean model $m_0(\x)=\ep (A\mid Y=0,\X=\x)$ and proposed a doubly robust estimating equation:
\begin{equation}
\frac{1}{n}\sum_{i=1}^n\widehat\phi(\X_i)\left\{Y_ie^{-\beta A_i-\X_i\trans\widehat\bgamma}-(1-Y_i)\right\}\left\{A_i-g(\X_i\trans\widehat\balpha)\right\}=0,
\label{equ:dr:par}
\end{equation}
where $\widehat\phi(\x)$ is an estimation of some scalar nuisance function $\phi(\x)$ affecting the asymptotic efficiency of the estimator, and $\widehat\balpha$ and $\widehat\bgamma$ are two fixed dimensional nuisance model estimators. Estimator $\widehat\beta$ solved from (\ref{equ:dr:par}) is doubly robust that it is valid when either $r(\x)=\x\trans\bgamma$ is correctly specified for the logistic model nonparametric component, or $m(\x)=g(\x\trans\balpha)$ is correctly specified for the conditional mean model $m_0(\x)=\ep (A\mid Y=0,\X=\x)$. Prior to this, the doubly robust semiparametric estimation of odds ratio was built upon $p(A\mid \X,Y=0)$, the conditional density of $A$ given $\X$ and $Y=0$ \citep[e.g.]{yun2007semiparametric,tchetgen2010doubly}, requiring a stronger model assumption\footnote{\cite{yun2007semiparametric} and \cite{tchetgen2010doubly} estimate the conditional distribution of $A$ given $\X$ while \cite{tan2019doubly} only needs to specify a conditional mean model of $A$.} than \cite{tan2019doubly} for continuous $A$.

Nevertheless, \cite{tan2019doubly} focuses on fixed dimensional parametric nuisance models that are still prone to model misspecification. And their proposed approach is not readily applicable to the recently developed and exploited HD ($p\gg n$ and the two nuisance components are specified as parametric models with sparse coefficients) or ML (the nuisance functions are estimated by arbitrary blackbox learning algorithm of condition mean) approaches. For the HD realization, this is because simply using some regularized nuisance estimators in (\ref{equ:dr:par}) would typically incur excessive bias and not guarantee the parametric rate of convergence. We realize bias reduction with respect to the regularization errors by constructing certain dantzig moment equations to estimate the nuisance parameters. With the ultra-sparse nuisance parameters, i.e. their sparsity level is $o(n^{1/2}/\log p)$, our proposed estimator preserves the model double robustness property that it approaches $\beta_0$ at the rate $O_p(n^{-1/2})$ when either $r(\cdot)$ or $m(\cdot)$ is correctly specified. Under the ML framework, the non-linearity and unextractablity of the logit link makes it impossible to naturally estimate $r_0(\cdot)$ with a learning algorithm of conditional mean\footnote{Without purposed modification, the natural form of most contemporary supervised learning methods, e.g. random forest, support vector machine and neural network, can be conceptualized as a blackbox algorithm of conditional mean estimation because their task is prediction for continuous responses and classification for categorical responses.} as trivially done under the partially linear model \citep{chernozhukov2016double}. We handle this challenge through an easy-to-implement ``full model refitting" (FMR) procedure that facilitates flexible implementation of arbitrary conditional mean learning algorithms in our framework. And our DML estimator for $\beta_0$ is rate doubly robust in the same sense as \cite{chernozhukov2016double}, i.e. being asymptotically normal at rate $n^{-1/2}$ when the two nuisance ML estimators are consistent for the true models and their mean squared errors are controlled by $o_p(n^{-1/2})$.

In recent years, there has been a large body of literature developed for semiparametric inference (of a low-dimensional parameter) with HD and ML nuisance models, which gets arising attention and application in economics and policy sciences \citep{athey2017state,knaus2018double,knaus2020double,yang2020double}. Similar to the structure of our paper, there are two different frameworks among the literature, HD setting and ML setting, though sometimes both of them referred as ``machine learning" approaches. As the main difference between them, the HD setting imposes parametric assumptions on the nuisance models and may allow for their potential misspecification while ML setting uses nonparametric ML estimation supposed to approach the true nuisance models at certain geometric rate. 

To estimate low dimensional parameters like average treatment effect (ATE) and conditionl treatment effect in linear or log-linear model under the HD setting with potentially misspecified nuisance models, recent work including \cite{smucler2019unifying,tan2020model,tan2020regularized,ning2020robust,dukes2020inference} constructed $\ell_1$-regularized estimating equations with certain $\ell_{\infty}$-constraints to simultaneously estimate the nuisance parameters and calibrate their first order bias. 
In comparison, \cite{bradic2019sparsity} proposed a more sparsity-rate robust ATE estimator that requires substantially weaker sparsity assumptions but needs both HD parametric nuisance models to be correctly specified. Among these literature, our work is the first to consider the logistic partially linear model under a similar regime. In parallel to these but closely relevant to our target, existing approaches including debiased (desparsified) LASSO \citep{van2014asymptotically,jankova2016confidence} and regularized Riesz representer \citep{chernozhukov2018double,belloni2018high} used the empirical inverse of the information matrix obtained with $\ell_1$-regularized regression to correct for the bias of logistic LASSO estimator. They imposed a technical ultra-sparsity condition on inverse of the information matrix, which has been criticized as unreasonable and unverifiable \citep{xia2020revisit}. Compared with them, our sparsity assumption is model-specific and thus more reasonable and explainable as will be shown later.

We note that near the finishing of our paper, a parallel paper, \cite{ghosh2020doubly} is published on arxiv\footnote{As a preliminary version of our paper, our previous technical note \url{https://arxiv.org/abs/2008.12467} appears slightly earlier than \cite{ghosh2020doubly}.}. Compared with their work, our HD part is studying the same problem but using different doubly robust estimating equation and calibrated procedures for the nuisance models. While the two proposals have similar theoretical properties and numerical implementation strategy (see our Sections \ref{sec:method:hd}, \ref{sec:thm:hd} and Appendix \ref{sec:app:num}). Also, they should have similar numerical performance. Certainly, our method under the ML setting introduced below has no overlapping with their work.

For the nonparametric ML realization, \cite{chernozhukov2016double} established a double machine learning (DML) framework utilizing Neyman orthogonal scores and cross-fitting to construct parametrically efficient ML based casual estimator. Their framework has been playing a central role in semiparametric inference with ML. As complements or extensions of it, recent work including  \cite{chernozhukov2018doublelocal,zimmert2019nonparametric,colangelo2020double} localized the orthogonal score function to estimate conditional average treatment effect; \cite{semenova2020debiased} constructed the best linear approximation of a structural function with ML; \cite{farrell2018deep} used deep neural networks for DML estimation; \cite{wager2018estimation,oprescu2019orthogonal} proposed and studied the tree-based ML approaches for causal inference; and \cite{cui2019bias} proposed a minimax data-driven model selection approach to choose the ML nuisance models with the lowest bias on the DML estimator. The above mentioned work elaborated on specific inference problems including partially linear model, ATE, and heterogeneous treatment effect with their nuisance models directly estimable with arbitrary (supervised) ML algorithms. While as summarized above, $r_0(\cdot)$ cannot be estimated likewise with general ML algorithms, due to the non-linear structure of (\ref{model}). To the best of our knowledge, this paper is the first one to solve this non-trivial technical problem through the proposed FMR procedure.

The rest of this paper is organized as follows. In Section \ref{sec:moti}, we define the Neyman orthogonal (doubly robust) score equation for logistic partially linear model. In Section \ref{sec:method}, we introduce the realization of debiased and DML inference for $\beta_0$ under the HD and ML settings separately. In Section \ref{sec:thm}, we present and justify the asymptotic property of our HD and ML estimator separately. In Sections \ref{sec:sim} and \ref{sec:real}, we conduct simulation studies on empirical performance of our method and apply it to assess the effect of EC pill on early gestation foetal.

\section{Neyman orthogonal score}\label{sec:moti}

Before coming to the specific approaches in Section \ref{sec:method}, we introduce a Neyman orthogonal (doubly robust) score function for logistic partially linear model and derive its first order bias, which plays a central role in motivating and guiding our method construction and theoretical analysis. Let observation $\D_i=\{Y_i,A_i,\X_i\}$ for $i=1,\ldots,n$ and $\D=\{Y,A,\X\}$ be a realization of $\D_i$. Motivated by equation (\ref{equ:dr:par}) proposed and studied in \cite{tan2019doubly}, we define the Neyman orthogonal score as
\[
h(\D;\beta,\eta)=\psi(\X)\{Ye^{-\beta A}-(1-Y)e^{r(\X)}\}\{A-m(\X)\},
\]
where $\eta=\{r(\cdot),m(\cdot),\psi(\cdot)\}$ represents the whole set of nuisance functions. In analog to (\ref{equ:dr:par}), $r(\cdot)$ and $m(\cdot)$ corresponds to the nonparametric component $r_0(\cdot)$ defined in (\ref{model}) and the nuisance model $m_0(\x)=\ep (A\mid Y=0,\X=\x)$, respectively. And $\psi(\x)$ is a nuisance function affecting the asymptotic variance of the estimator that may depend on $r(\x)$ and $m(\x)$ and actually corresponds to $\phi(\x)e^{-r(\x)}$ with $\phi(\x)$ defined by (\ref{equ:dr:par})\footnote{We rewrite (\ref{equ:dr:par}) with $\psi(\x)=\phi(\x)e^{-r(\x)}$ to induce the form of score function with both its partial derivatives on $r$ and $\psi$ are Neyman orthogonal.}. 

\begin{remark}
\label{rem:2.1}
According to \cite{tan2019doubly}, the score function $h(\D;\beta,\eta)$ is doubly robust in the sense that when $r(\cdot)=r_0(\cdot)$ or $m(\cdot)=m_0(\cdot)$, $\beta_0$ solves $\ep h(\D;\beta,\eta)=0$. We shortly demonstrate this as follows. When either $r(\cdot)=r_0(\cdot)$ or $m(\cdot)=m_0(\cdot)$ holds, we have
\begin{align*}
&\ep\psi(\X)(1-Y)\{e^{r(\X)}-e^{r_0(\X)}\}\{A-m(\X)\}\\
=&\ep\left[\psi(\X)\{e^{r(\X)}-e^{r_0(\X)}\}\{A-m(\X)\}\Big| Y=0,\X\right]=0,
\end{align*}
which combined with (\ref{model}) lead to that
\begin{align*}
\ep h(\D;\beta_0,\eta)=&\ep\psi(\X)\{Ye^{-\beta_0 A}-(1-Y)e^{r(\X)}\}\{A-m(\X)\}\\
=&\ep\psi(\X)e^{r_0(\X)}\{Ye^{-\beta_0 A-r_0(\X)}-(1-Y)\}\{A-m(\X)\}\\
=&\ep\psi(\X)e^{r_0(\X)}\{A-m(\X)\}\frac{Y-\Pbb (Y=1\mid A,\X)}{\Pbb(Y=1\mid A,\X)}=0.
\end{align*}
\end{remark}
Suppose the nuisance models $r_0(\x)$ and $m_0(\x)$ are estimated by $\widehat r(\x)$ and $\widehat m(\x)$ converging to $\bar r(\x)$ and $\bar m(\x)$ and $\widehat\psi(\x)$ represents the estimator for $\psi(\x)$ approaching $\bar \psi(\x)$. Denote by $\bar\eta=\{\bar r(\cdot),\bar m(\cdot),\bar\psi(\cdot)\}$ and $\widehat\eta=\{\widehat r(\cdot),\widehat m(\cdot),\widehat\psi(\cdot)\}$. We then write the Gateaux (partial) derivative of the score function $h(\D;\beta_0,\bar\eta)$ as 
\begin{equation}
\begin{split}
&\partial_{\eta} h(\D;\beta_0,\bar\eta)[\eta-\bar\eta]\\
=&\partial_{\psi} h(\D;\beta_0,\bar\eta)[\psi-\bar\psi]+\partial_{r} h(\D;\beta_0,\bar\eta)[r-\bar r]+\partial_{m} h(\D;\beta_0,\bar\eta)[m-\bar m]\\
=:&\{Ye^{-\beta_0 A}-(1-Y)e^{\bar r(\X)}\}\{A-\bar m(\X)\}\{\psi(\X)-\bar\psi(\X)\}\\
&-(1-Y)\bar\psi(\X)e^{\bar r(\X)}\{A-\bar m(\X)\}\{r(\X)-\bar r(\X)\}\\
&-\bar\psi(\X)\{Ye^{-\beta_0 A}-(1-Y)e^{\bar r(\X)}\}\{m(\X)-\bar m(\X)\}.
\end{split}
\label{equ:derive}    
\end{equation}
\begin{remark}
\label{rem:2.2}
We evaluate the Neyman orthogonal score on some limiting parameters $\bar r(\cdot)$ and $\bar m(\cdot)$ instead of on $r_0(\cdot)$ and $m_0(\cdot)$ as in \cite{chernozhukov2016double}. This is because different from their ML framework (and ours) assuming both nuisance estimators converge to the true models, i.e. $\bar r(\cdot)=r_0(\cdot)$ and $\bar m(\cdot)=m_0(\cdot)$, our HD realization allows at most one nuisance model to be wrongly specified, and thus the score function to be analyzed with may not be evaluated at the true models.
\end{remark}
Inspired by our deduction in Remark \ref{rem:2.1}, $\ep\partial_{\psi} h(\D;\beta_0,\bar\eta)[\psi-\bar\psi]=0$ for any $\psi$ whenever $\bar r(\cdot)=r_0(\cdot)$ or $\bar m(\cdot)=m_0(\cdot)$. Also, $\ep\partial_{r} h(\D;\beta_0,\bar\eta)[r-\bar r]=0$ when $\bar m(\cdot)=m_0(\cdot)$ and $\ep\partial_m h(\D;\beta_0,\bar\eta)[m-\bar m]=0$ when $\bar r(\cdot)=r_0(\cdot)$. Thus, under the ML setting, $h(\D;\beta_0,\bar\eta)$ satisfies Neyman orthogonality defined in \cite{chernozhukov2016double} and the first order (over-fitting) bias of the estimating equation: $n^{-1}\sum_{i=1}^n h(\D_i;\beta,\widehat\eta)=0$ can be removed through cross-fitting (introduced in Section \ref{sec:method:ml}) and concentration. While under the HD parametric setting with $\bar r(\cdot)\neq r_0(\cdot)$ or $\bar m(\cdot)\neq m_0(\cdot)$, one needs to carefully construct the moment equations for $\bar r(\cdot)$ and $\bar m(\cdot)$ to ensure the orthogonality conditions. While similar to existing literature \citep{chernozhukov2016double,tan2020model}, removal of the second order (and beyond) bias relies on the assumptions that quality of the nuisance estimators $\widehat r(\cdot)$ and $\widehat m(\cdot)$ are good enough (see Section \ref{sec:thm}).

\section{Method}\label{sec:method}

In this section, we separately present our specific construction procedures for HD parametric and ML nonparametric realization of the debiased/DML estimator for $\beta_0$, based on the Neyman orthogonal score derived in Section \ref{sec:moti}.

\subsection{High dimensional parametric model realization}\label{sec:method:hd}

Consider the setting with $p\gg n$, each $\X_i$ has its first element being $1$, $r(\x)=\x\trans\bgamma$ and $m(\x)=g(\x\trans\balpha)$ where $g(\cdot)$ is a monotone and smooth link function with derivative $g'(\cdot)$. Inspired by existing work including \cite{smucler2019unifying,tan2020model,dukes2020inference}, we construct dantzig moment equations to ensure the Neyman orthogonality empirically: $\partial_{r} h(\D;\beta_0,\bar\eta)[r-\bar r]=0$ and $\partial_m h(\D;\beta_0,\bar\eta)[m-\bar m]=0$, under potential misspecification of (at most one) the nuisance models.

First, we obtain $\widetilde\bgamma$ as some initial estimator for $\bgamma$ that converges to some limiting parameter $\bgamma^*$ equaling to the true model parameter $\bgamma_0$ when $r(\x)$ is correctly specified: $r(\x)=\x\trans\bgamma_0$. And let $\widehat\psi(\x)$ be some estimator of the nuisance function $\psi(\x)$ depending on $\widetilde\bgamma$ with its limiting function being $\bar\psi(\x)$, whose choice will be discussed in Section \ref{sec:method:eff} later. According to (\ref{equ:derive}), we obtain $\widehat\balpha$ through the dantzig moment equation:
\begin{equation}
{\rm min}_{\balpha\in\mathbb{R}^p}\|\balpha\|_1\quad{\rm s.t}\quad\left\|n^{-1}\sum_{i=1}^n(1-Y_i)\widehat\psi(\X_i)e^{\X_i\trans\widetilde\bgamma}\{A_i-g(\X_i\trans\balpha)\}\X_i\right\|_{\infty}\leq\lambda_{\alpha},
\label{equ:dant:m}    
\end{equation}
where $\lambda_{\alpha}$ is a tuning parameter controlling the regularization bias. Finally, we obtain the nuisance estimator $\widehat\bgamma$ and the targeted HD estimator $\widehat\beta\subhd$ simultaneously from:
\begin{equation}
\label{equ:dant:r}   
\begin{split}
{\rm min}_{\beta\in\mathbb{R},\bgamma\in\mathbb{R}^p}\|\bgamma\|_1\quad{\rm s.t}\quad\left\|n^{-1}\sum_{i=1}^n\widehat\psi(\X_i)\{Y_ie^{-\beta A_i}-(1-Y_i)e^{\X_i\trans\bgamma}\}g'(\X_i\trans\widehat\balpha)\X_i\right\|_{\infty}&\leq\lambda_{\gamma};\\
n^{-1}\sum_{i=1}^n\widehat\psi(\X_i)\{Y_ie^{-\beta A_i}-(1-Y_i)e^{\X_i\trans\bgamma}\}\{A_i-g(\X_i\trans\widehat\balpha)\}&=0.
\end{split}
\end{equation}
Let the limits of $\{\widehat\balpha,\widehat\bgamma\}$ be $\{\bar\balpha,\bar\bgamma\}$, and $\bar\eta=\{\bar r(\cdot),\bar m(\cdot),\bar\psi(\cdot)\}$ where $\bar r(\x)=\x\trans\bar\bgamma$, $\bar m(\x)=g(\x\trans\bar\balpha)$ and $\bar\psi(\x)$ as given in Section \ref{sec:method:eff}. We shall comment on the orthogonality (moment) conditions of our proposal in Remark \ref{rem:3.0}, compare our method with \cite{dukes2020inference} in Remark \ref{rem:3.2}, and discuss its numerical implementation with a weighted LASSO formation in Remark \ref{rem:3.1}.

\begin{remark}
\label{rem:3.0}
Neglect the second (and above) order error terms for now. When $r(\x)$ is correct (see Assumption \ref{asu:hd:1}), i.e. $r_0(\x)=\x\trans\bgamma_0$, it naturally holds that $\ep\partial_m h(\D;\beta_0,\bar\eta)[\widehat m-m_0]=0$ and $\bgamma^*=\bar\bgamma=\bgamma_0$. Then the $\ell_{\infty}$-constraint in (\ref{equ:dant:m}) imposes that
\[
\ep\partial_{r} h(\D;\beta_0,\bar\eta)[\widehat r-r_0]\approx\ep(1-Y)\bar\psi(\X)e^{\X\trans\bgamma_0}\{A-g(\X\trans\bar\balpha)\}\X\trans(\bgammahat-\bgamma_0)=\bzero\trans(\bgammahat-\bgamma_0).
\]
When $\bar m(\x)=m_0(\x)=g(\x\trans\balpha_0)$, we have $\ep\partial_{r} h(\D;\beta_0,\bar\eta)[\widehat r-r_0]$ and $\bar\balpha=\balpha_0$ in turn. And the $\ell_{\infty}$-constraint of (\ref{equ:dant:r}) corresponds to
\[
\ep\partial_m h(\D;\beta_0,\bar\eta)[\widehat m-m_0]\approx\ep\bar\psi(\X)\{Ye^{-\beta_0 A}-(1-Y)e^{\X\trans\bar\bgamma}\}g'(\X_i\trans\balpha_0)\X\trans(\widehat\balpha-\balpha_0)=\bzero\trans(\widehat\balpha-\balpha_0).
\]
Thus, the Neyman orthogonality condition $\partial_{\eta} h(\D;\beta_0,\bar\eta)[\eta-\bar\eta]=0$ as introduced in Section \ref{sec:moti} is satisfied under our construction when either $r(\cdot)$ or $m(\cdot)$ is correctly specified.
\end{remark}

\begin{remark}
\label{rem:3.2}
Similar the HD partially linear (or log-linear) setting studied in \cite{dukes2020inference}, estimating equation for the nuisance parameter $\bgamma$ in our framework involves the unknown $\beta$. Unlike their construction procedure that plug-in $\beta_0$ as every $\beta\in\mathbb{R}$ to estimate $\bgamma$ and invert the resulted score-test $p$-values for interval estimation of $\beta_0$, we solve for $\widehat\beta$ and $\widehat\bgamma$ jointly from (\ref{equ:dant:r}) with the moment equation for $\widehat\bbeta$ being doubly robust, as demonstrated in Remark \ref{rem:2.1}. Compared with their method, ours is more friendly in computation and implementation, additionally provides $n^{-1/2}$-consistent point estimator of $\beta_0$ and preserves similar theoretical guarantee (see Section \ref{sec:thm:hd}).
\end{remark}

\begin{remark}
\label{rem:3.1}
As is detailed in Appendix \ref{sec:app:num}, one could also construct LASSO problems with the same Karus--Kuhn--Tucker (KKT) conditions as the $\ell_{\infty}$-norm constraints in (\ref{equ:dant:m}) and (\ref{equ:dant:r}), to obtain the estimators $\widehat\balpha$ and $\widehat\bgamma$, which has equivalent theoretical properties as the dantzig equations.\footnote{Here we present the dantzig equation version because it is more intuitive and directly connected with the Neyman orthogonal conditions.} Due to the non-convexity\footnote{Its partial derivative on $\beta$, $n^{-1}\sum_{i=1}^n-\widehat\psi(\X_i)Y_ie^{-\beta A_i}A_i\{A_i-g(\X_i\trans\widehat\balpha)\}$ is not always positive or negative definite, through empirically and theoretically, the partial derivative should be negative with very high chances.} of the equation in the second row of (\ref{equ:dant:r}), numerical solution of the LASSO counterpart of (\ref{equ:dant:r}) cannot be obtained with existing software like {\bf R} package ``glmnet" \citep{friedman2010regular} and ``RCAL" \citep{tan2019rcal}. A direct solution to this is programming an optimization procedure such as the Fisher scoring descent algorithm used by \cite{tan2020regularized}. While we also find a convenient way that moderately modifies the construction procedure to make the regularized estimating equations solvable with {\bf R} package ``RCAL", and use it for the numerical implementation of our method. In Appendix \ref{sec:app:num}, we outline this modification and demonstrate its theoretical guarantee.
\end{remark}

\subsection{Machine learning realization}\label{sec:method:ml}

We turn to a (nonparametric) ML setting under which any learning algorithms of conditional mean could potentially be applied to estimate the nuisance functions. Similar to \cite{chernozhukov2016double}, we randomly split the $n$ samples into $K$ folds: $\Isc_1,\Isc_2,\ldots,\Isc_K$ of equal size, to assist removing the first order (over-fitting) bias through concentration. Then the cross-fitted estimating equation for $\beta$ is constructed as
\begin{equation}
n^{-1}\sum_{k=1}^K\sum_{i\in\Isc_k} h(\D_i;\beta,\widehat\eta\supmfk)=0,
\label{equ:dr:cross}     
\end{equation}
where $\widehat\eta\supmfk=\{\widehat r\supmfk(\cdot),m\supmfk(\cdot),\psi\supmfk(\cdot)\}$, representing ML estimators converging to $\bar r(\cdot)=r_0(\cdot)$, $\bar m(\cdot)=m_0(\cdot)$ and $\bar\psi(\cdot)$, obtained with the samples in $\Isc_{\text{-}k}=\{1,\ldots,n\}\setminus\Isc_k$ and independent of the samples in $\Isc_k$. Now we present the specific construction procedures for $\widehat r\supmfk(\cdot)$ and $m\supmfk(\cdot)$ with the choice of $\psi\supmfk(\cdot)$ again discussed in Section \ref{sec:method:eff}.

Suppose there is a blackbox learning algorithm $\Lscr(R_i,\bC_i;\Isc)$ that inputs samples $\Isc\subseteq\{1,2,\ldots,n\}$ with some response $R_i$ and covariates $\bC_i$, and outputs an estimation of $\ep[R_i\mid\bC_i,i\in\Isc]$. We outline as follows our approach utilizing $\Lscr$ to estimate the nuisance functions. Corresponding to the definition of $m_0(\cdot)$, it can be estimated by: $\widehat m\supmfk(\cdot)=\mathscr{L}(A_i,\X_i;\Isc_{\text{-}k}\cap\{i:Y_i=0\})$. Compared to the partially linear setting in \cite{chernozhukov2016double}, estimation of $r_0(\cdot)$ with $\mathscr{L}$ is more sophisticated since it is defined through a non-linear form: $\Pbb (Y=1\mid A,\X)={\rm expit}\{\beta_0 A+r_0(\X)\}$. One could modify some ML approaches, e.g. neural network\footnote{By setting the last layer of the neural network to be the combination of a complex network of $\X$ and a linear function of $A$ and linking it with the outcome through an expit link.} to accommodate this form. However, such modification is not readily available in general, and typically requires additional human efforts on its implementing and validating if it exists. So alternatively, we propose a ``full model refitting" (FMR) procedure using an arbitrary $\mathscr{L}$ to estimate $r_0(\cdot)$. Our method is motivated by a simple proposition:
\begin{prop}
\label{prop:1}
Let $M_0(A,\X)=\Pbb (Y=1\mid A,\X)={\rm expit}\{\beta_0 A+r_0(\X)\}$. We have:
\[
\beta_0={\rm argmin}_{\beta\in\mathbb{R}} \ep\left[{\rm logit}\{M_0(A,\X)\}-\beta(A-\ep [A|\X])\right]^2.
\]
\end{prop}
\begin{proof}
For any $\beta\in\mathbb{R}$, we have
\begin{align*}
&\ep \left[{\rm logit}\{M_0(A,\X) \}-\beta(A-\ep [A|\X])\right]^2=\ep \left\{\beta_0 A+r_0(\X)-\beta(A-\ep [A|\X])\right\}^2\\
=&\ep \left\{(\beta_0-\beta)(A-\ep [A|\X])+\eta(\X)\right\}^2=(\beta_0-\beta)^2\ep (A-\ep [A|\X])^2+\ep \{\eta(\X)\}^2,
\end{align*}
where $\eta(\X)=r_0(\X)+\beta_0\ep [A|\X]$. Thus, $\beta_0$ minimizes $\ep \left[{\rm logit}\{M_0(A,\X) \}-\beta(A-\ep [A|\X])\right]^2$.
\end{proof}
Further randomly each split $\Isc_{\text{-}k}$ into $K$ folds $\Isc_{\text{-}k,1},\ldots\Isc_{\text{-}k,K},$ of equal size and denote by $\Isc\submfkj=\Isc_{\text{-}k}\setminus\Isc_{\text{-}k,j}$. Motivated by Proposition \ref{prop:1}, we first estimate the ``full" model $M_0(A,\X) $ with $\Isc_{\text{-}k,\text{-}j}$ as: 
\[
\widehat M\supmfkj(\cdot)=\mathscr{L}(Y_i,(A_i,\X_i\trans)\trans;\Isc\submfkj),
\]
and learn $a_0(\x)=\ep[A|\X=\x]$ as $\widehat a\supmfkj(\cdot)=\mathscr{L}(A_i,\X_i;\Isc\submfkj)$. Then we fit the (cross-fitted) least square regression to obtain:
\begin{equation}
\breve\beta\supmfk={\rm argmin}_{\beta\in\mathbb{R}}\frac{1}{|\Isc_{\text{-}k}|}\sum_{j=1}^K\sum_{i\in\Isc_{\text{-}k,j}}\left[{\rm logit}\{\widehat M\supmfkj(A_i,\X_i)\}-\beta\{A_i-\widehat a\supmfkj(\X_i)\}\right]^2,
\label{equ:est:init}
\end{equation}
as an estimator approaching $\beta_0$ at certain rate typically larger than $n^{-1/2}$. Then $r_0(\cdot)$ could be identified through $r_0(\X_i)={\rm logit}\{M_0(A_i,\X_i)\}-\beta_0 A_i$. Note that the empirically estimated version of ${\rm logit}\{M_0(A_i,\X_i)\}-\beta_0 A_i$ typically involves $A_i$ due to the discrepancy of $\beta_0$ and $M(\cdot)$ from their empirical estimation. This can essentially impede removal of the over-fitting bias since $\partial_{r} h(\D;\beta_0,\bar\eta)$ is not orthogonal to the error functions dependending on $A$. So we further estimate the conditional mean of ${\rm logit}\{M_0(A,\X) \}-\beta_0 A$ on $\X$ to estimate $r_0(\cdot)$. Denote by $W_i={\rm logit}\{\widehat M\supmfkj(A_i,\X_i)\}$ for each $i\in\Isc_{\text{-}k,j}$ and get $\widehat t\supmfk(\cdot)=\mathscr{L}(W_i,\X_i;\Isc_{\text{-}k})$ to estimate $t_0(\x)=:\ep [{\rm logit}\{M_0(A,\X) \}|\X=\x]$. Then the estimator of $r_0(\cdot)$ is given by:
\begin{equation}
\widehat r\supmfk(\x)=\widehat t\supmfk(\x)-\breve\beta\supmfk \widehat a\supmfk(\x),\quad\mbox{where}\quad\widehat a\supmfk(\x)=\frac{1}{K}\sum_{j=1}^K\widehat a\supmfkj(\x).
\label{equ:est:r:1}
\end{equation}
Alternatively, one can estimate $r_0(\cdot)$ through
\begin{equation}
\widehat r\supmfk(\cdot)=\log\left(\frac{\mathscr{L}(e^{-\breve\beta\supmfk A_i},\X_i;\Isc_{\text{-}k}\cap\{i:Y_i=1\})}{\mathscr{L}(1-Y_i,\X_i;\Isc_{\text{-}k})}\right),
\label{equ:est:r:2}
\end{equation}
inspired by the moment condition that is sufficient to identify $r_0(\cdot)$:
\[
\ep \left[Ye^{-\beta_0 A}-(1-Y)e^{r_0(\X)}\Big|\X\right]=\ep \left[e^{-\beta_0 A}\Big|\X,Y=1\right]-e^{r_0(\X)}\ep \left[(1-Y)\big|\X\right]=0.
\]
We refer the estimation step for $\breve\beta\supmfk$ and $\widehat r\supmfk(\cdot)$ introduced above as ``refitting", and the whole procedure as FMR since we ``refit" the least square problem (\ref{equ:est:init}) and ML models $\Lscr$ to estimate $r_0(\cdot)$ with the initially estimated full model ${\rm logit}\{M_0(A_i,\X_i)\}$ as a psuedo-outcome. Finally, we solve (\ref{equ:dr:cross}) based on $\widehat\eta\supmfk$ to obtain the DML estimator $\widehat\beta\subml$.
\begin{remark}
We further use cross-fitting in FMR to avoid over-fitting of the models $\widehat M\supmfkj(\cdot)$ and $a\supmfkj(\cdot)$ when they are used to obtain the estimators $\breve\beta\supmfk$, $\widehat t\supmfk(\x)$ and $\widehat r\supmfk(\x)$.
\label{rem:3.3}
\end{remark}

\begin{remark}
The FMR implicitly assumes that $\Lscr$ should perform similarly well on different learning objects with the covariates set as either $\X$ or $(A,\X\trans)\trans$. Classic nonparametric approaches like kernel smoothing or sieve may not satisfy this since including one more covariate $A$ in addition to the very low dimensional $\X$ can have substantial impact on estimation performance. Thus, we recommend using more dimensionality-robust modern ML approaches, such as random forest and neural network, in our ML framework. While the classic ``plug-in" sieve or kernel method has been well-studied in existing literature \citep{severini1994quasi,lin2006semiparametric}.
\end{remark}

\subsection{Efficiency consideration}\label{sec:method:eff}

The nuisance function $\psi(\cdot)$ in our framework is included and chosen in consideration of estimation efficiency. \cite{tan2019doubly} proposed and studied two options for $\phi(\cdot)$ used and defined in (\ref{equ:dr:par}), with the corresponding function $\psi(\cdot)$ taken as:
\[
\psi_{\rm opt}(\x)=\frac{e^{-r(\x)}\ep [\{A-m(\X)\}^2|\X=\x,Y=0]}{\ep [\{A-m(\X)\}^2/{\rm expit}\{\beta_0 A+ r(\X)\}|\X=\x,Y=0]};\quad\psi_{\rm simp}(\x)={\rm expit}\{-r(\x)\}.
\]
It was shown that when both nuisance models are correctly specified, the estimator solved with the weight $\psi_{\rm opt}(\cdot)$ achieves the minimum asymptotic variance. However, computation of $\psi_{\rm opt}(\cdot)$ involves numerical integration with respect to $\X$ given $Y=0$, making it sometimes inconvenient to implement. So \cite{tan2019doubly} proposed a simplified but reasonable choice $\psi_{\rm simp}(\x)$ defined as above that is obtained by evaluating $\psi_{\rm opt}(\x)$ at $\beta_0=0$. In the following theoretical and numerical studies, we stick to $\psi(\x)=\psi_{\rm simp}(\x)$, $\widehat\psi(\x)={\rm expit}(-\x\trans\widetilde\bgamma)$ and correspondingly $\bar\psi(\x)={\rm expit}(-\x\trans\bgamma^*)$ under the HD setting, and $\widehat\psi\supmfk(\x)={\rm expit}\{-\widehat r\supmfk(\x)\}$ and $\bar\psi(\x)={\rm expit}\{-r_0(\x)\}$ under the ML setting.

\section{Asymptotic analysis}\label{sec:thm}

Let $o(\alpha_n)$, $O(\alpha_n)$, $\omega(\alpha_n)$, $\Omega(\alpha_n)$ and $\Theta(\alpha_n)$ represent the sequences growing at a smaller, equal/smaller, larger, equal/larger and equal rate of $\alpha_n$, respectively. And let $o_{\Pbb}$, $O_{\Pbb}$, $\omega_{\Pbb}$, $\Omega_{\Pbb}$ and $\Theta_{\Pbb}$ be the corresponding rates with probability approaching $1$ as $n\rightarrow \infty$. Let $\mathcal{X}\subseteq\mathbb{R}^p$ be the domain of $\X$. First, we introduce the regularity condition for $\beta$ and its estimating equation used under both HD and ML settings as Assumption \ref{asu:1}, which is standard and can be commonly found in literature of the asymptotic analysis of $M$-estimator \citep[Chapter 5]{van2000asymptotic}. And we shall then study the asymptotic properties of $\widehat\beta\subhd$ and $\widehat\beta\subml$ in Sections \ref{sec:thm:hd} and \ref{sec:thm:ml} respectively.

\setcounter{assume}{0}
\renewcommand{\theassume}{REG}

\begin{assume}[Regularity of estimating equation for $\beta$]
\label{asu:1}
Parameter $\beta$ belongs to a compact set $\Bsc\subseteq\mathbb{R}$ and there exists $\delta_n=\Omega(n^{-1/2}\log n)$ such that $(\beta_0-\delta_n,\beta_0+\delta_n)\subseteq\Bsc$. Exposure $A$ belongs to a compact set $\mathcal{A}$ and $\sup_{\x\in\mathcal{X}}|\ep[A|\X=\x,Y=y]|=O(1)$ for $y=0,1$. Also, it is satisfied that 
\[
\ep\bar\psi(\X)Ye^{-\beta_0 A}A\{A-\bar m(\X)\}=\Theta(1)\quad\mbox{and}\quad\ep h^2(\D;\beta_0,\bar\eta)=\Theta(1).\footnote{To accommodate the notations of both HD and ML, we use $\bar m(\cdot)$ and $\bar r(\cdot)$ to represent the limiting models defined as $g(\x\trans\bar\balpha)$ and $\x\trans\bar\bgamma$ under the HD setting and just the true models $m_0(\cdot)$ and $r_0(\cdot)$ under ML.}
\]
\end{assume}

\subsection{High dimensional (parametric) setting}\label{sec:thm:hd}
Let ${\rm expit}'(\cdot)$ be the derivative function of ${\rm expit}(\cdot)$, $\|\cdot\|_0$ represents the number of non-zero elements in a vector and $s=\max\{\|\bgamma^*\|_0,\|\bar\bgamma\|_0,\|\bar\balpha\|_0\}$. We introduce following assumptions to regularize the covariates and nuisance estimators.

\setcounter{assume}{0}
\renewcommand{\theassume}{HD\arabic{assume}}

\begin{assume}[Model double robustness]
At least one of the following conditions hold: (a) there exists $\bgamma_0\in\mathbb{R}^p$ such that $r_0(\x)=\x\trans\bgamma_0$ and $\bgamma^*=\bar\bgamma=\bgamma_0$; (b) there exists $\balpha_0\in\mathbb{R}^p$ such that $m_0(\x)=g(\x\trans\balpha_0)$ and $\bar\balpha=\balpha_0$.
\label{asu:hd:1}
\end{assume}

\begin{assume}[Concentration rate]
It holds that
\begin{align*}
\left\|n^{-1}\sum_{i=1}^n(1-Y_i){\rm expit}(-\X_i\trans\bgamma^*)e^{\X_i\trans\bar\bgamma}\{A_i-g(\X_i\trans\bar\balpha)\}\X_i\right\|_{\infty}=O_{\Pbb}\{(\log p/n)^{1/2}\};
\end{align*}
\begin{align*}
&\left\|n^{-1}\sum_{i=1}^n{\rm expit}(-\X_i\trans\bgamma^*)\{Y_ie^{-\beta_0 A_i}-(1-Y_i)e^{\X_i\trans\bar\bgamma}\}g'(\X_i\trans\bar\balpha)\X_i\right\|_{\infty}=O_{\Pbb}\{(\log p/n)^{1/2}\};\\
&\left\|n^{-1}\sum_{i=1}^n{\rm expit}'(-\X_i\trans\bgamma^*)\{Y_ie^{-\beta_0 A_i}-(1-Y_i)e^{\X_i\trans\bar\bgamma}\}\{A_i-g(\X_i\trans\bar\balpha)\}\X_i\right\|_{\infty}=O_{\Pbb}\{(\log p/n)^{1/2}\}.
\end{align*}
\label{asu:hd:2}
\end{assume}

\begin{assume}[Smooth link function]
\label{asu:hd:3}
There exists $L=\Theta(1)$ that for any $u,v\in\mathbb{R}$,
\[
|g'(u)-g'(v)|\leq L|u-v|.
\]
\end{assume}

\begin{assume}[Risk of the $\ell_1$-regularized estimators]
\label{asu:hd:4}
There exists tuning parameters $\lambda_{\alpha},\lambda_{\gamma}=\Theta\{(\log p/n)^{1/2}\}$ such that (\ref{equ:dant:m}) and (\ref{equ:dant:r}) have feasible solutions with probability approaching $1$ and
\begin{align*}
&\sup_{i\in\{1,\ldots,n\}}|g(\X_i\trans\widehat\balpha)|=O_{\Pbb}(1);\quad\|\widetilde\bgamma-\bgamma^*\|_1+\|\widehat\bgamma-\bar\bgamma\|_1+\|\widehat\balpha-\bar\balpha\|_1=O_{\Pbb}\{s(\log p/n)^{1/2}\};\\
&n^{-1}\sum_{i=1}^n\{1+e^{\X_i\trans\bar\bgamma}\}\left[\{\X_i\trans(\widehat\bgamma-\bar\bgamma)\}^2+\{\X_i\trans(\widetilde\bgamma-\bgamma^*)\}^2\right]+(\widehat\beta\subhd-\beta_0)^2=O_{\Pbb}(s\log p/n);\\
&n^{-1}\sum_{i=1}^n\{1+e^{\X_i\trans\bar\bgamma}\}\left[\{\X_i\trans(\widehat\balpha-\bar\balpha)\}^2+\{g(\X_i\trans\widehat\balpha)-g(\X_i\trans\bar\balpha)\}^2\right]=O_{\Pbb}(s\log p/n).
\end{align*}
\end{assume}

\begin{assume}[Ultra-sparsity]
It holds that $s=o(n^{1/2}/\log p)$.
\label{asu:hd:5}
\end{assume}

\begin{remark}
\label{rem:4.1.1}
Under Assumption \ref{asu:hd:1} and our constructions (\ref{equ:dant:m}) and (\ref{equ:dant:r}) (or the one introduced in Appendix \ref{sec:app:num}), the expectations of the terms to be concentrated in Assumption \ref{asu:hd:2} are $\bzero$ by Remark \ref{rem:3.0}. Then their maximum norms can be controlled by $O_{\Pbb}\{(\log p/n)^{1/2}\}$ as assumed in \ref{asu:hd:2}, when the covariates $\X_i$ are bounded, subgaussian or beyond \citep{kuchibhotla2018moving}, using the concentration results derived in existing literature \citep{gine2016mathematical}. 
\end{remark}

\begin{remark}
\label{rem:4.1.2}
Rates of the prediction and estimation risk of the nuisance estimators in Assumption \ref{asu:hd:4} can be derived following the general theoretical framework for $\ell_1$-regularized estimation introduced in \cite{candes2007dantzig,bickel2009simultaneous,buhlmann2011statistics,negahban2012unified}. The same rate properties has been used for analyzing doubly robust estimators with HD nuisance models in existing literature \citep{smucler2019unifying,tan2020model,dukes2020inference}. Note that (\ref{equ:dant:m}) and (\ref{equ:dant:r}) involves the estimators $\widetilde\bgamma$ or $\widehat\balpha$ obtained beforehand. This will require some additional effort on removing the  ``plug-in" errors of $\widetilde\bgamma$ or $\widehat\balpha$ when deriving the risk rates for $\widehat\balpha$ and $\widehat\bgamma$, compared to the standard analysis procedures. One could see \cite{tan2020model} for a similar technical issue and relevant technical details being used to handle it. 

In addition, $\sup_{i\in\{1,\ldots,n\}}|g(\X_i\trans\widehat\balpha)|=O_{\Pbb}(1)$ imposed in \ref{asu:hd:4} is not a standard assumption but is mild and very likely to hold: $\sup_{\x\in\mathcal{X}}|g(\x\trans\bar\balpha)|=O(1)$ according to Assumption \ref{asu:1} so we only need $g(\x\trans\widehat\balpha)-g(\x\trans\bar\balpha)$ to be $O_{\Pbb}(1)$ uniformly.

\end{remark}

\begin{remark}
\label{rem:4.1.3}
The ultra-sparsity assumption \ref{asu:hd:5} was also imposed in existing literature including \cite{tan2020model} and \cite{dukes2020inference}, to control the rate of bias incurred by the HD estimators: $s\log p/n$ below the parametric rate. For linear nuisance model, existing work like \cite{zhu2018significance} and \cite{dukes2020inference} suggested to add additional moment (KKT) constraints to relax the ultra-sparsity assumption. However, their approach has not been shown to be feasible for the case with non-linear models yet, so it may be promising but still remains unclear for our framework.
\end{remark}

We present the asymptotic property of $\widehat\beta\subhd$ in Theorem \ref{thm:2} and its proof in Appendix \ref{sec:app:thm}.
\begin{theorem}
Denote by $\bar I=\ep\bar\psi(\X)Ye^{-\beta_0 A}A\{A-\bar m(\X)\}$ and $\bar\sigma^2=\bar I^{-2}\ep h^2(\D;\beta_0,\bar\eta)$. Under Assumptions \ref{asu:1} and \ref{asu:hd:1}--\ref{asu:hd:5}, we have
\[
\sqrt{n}\bar\sigma^{-1}(\widehat\beta\subhd-\beta_0)=\frac{1}{\sqrt{n}}\sum_{i=1}^n(\bar\sigma\bar I)^{-1}h(\D_i;\beta_0,\bar\eta)+o_{\Pbb}(1),
\]
which weakly converge to ${\rm N}(0,1)$.
\label{thm:1}
\end{theorem}

Recently, logistic debiased LASSO \citep{van2014asymptotically,jankova2016confidence} has been criticized on that its sparse inverse information matrix condition is not explainable and justifiable, leading to a subpar performance theoretically and numerically \citep{xia2020revisit}. Interestingly, we find the model sparsity assumption of our method is more reasonable than debiased LASSO and present a simple comparison of these two approaches in Remark \ref{rem:4.1}.

\begin{remark}
Assume the logistic model $\Pbb (Y=1\mid A,\X)={\rm expit}\{\beta_0 A+\X\trans\bgamma_0\}$ is correctly specified. As is argued by \cite{xia2020revisit}, assuming the information matrix of the logistic model has an ultra-sparse{\footnote{Or approximately sparse (see recent work like \cite{belloni2018high,ma2020global,liu2020integrative}).}} inverse is crucial to ensure the desirable parametric rate of the debiased LASSO estimator for $\beta_0$. However, this assumption is not explainable or convincing for the common gaussian design with sparse precision matrix, due to the presence of logistic canonical link. In comparison, we require that $\ep (A\mid Y=0,\X=\x)=g(\X_i\trans\balpha_0)$ with $\|\balpha_0\|_0=o(n^{1/2}/\log p)$. This assumption has two advantages. First, it accommodates nonlinear link function $g(\cdot)$ and can be more reasonable for a categorical $A$. Second, it is imposed on a conditional model directly and thus more explainable. For example, consider a conditional gaussian model: $(A,\X\trans)\trans\mid \{Y=j\}\sim \mathcal{N}(\bmu_j,\boldsymbol{\Sigma})$ for $j=0,1$. Then we have $r_0(\x)=\x\trans\bgamma_0$ where $(\beta_0,\bgamma_0\trans)\trans=\boldsymbol{\Sigma}^{-1}(\bmu_1-\bmu_0)$ and $A\mid \X,Y=0$ follows a gaussian linear model with the coefficient $\balpha_0$ determined by $\boldsymbol{\Sigma}^{-1}$. Therefore, our sparsity assumptions on $\balpha_0$ and $\bgamma_0$ actually assumes the data generation parameters $\boldsymbol{\Sigma}^{-1}$ and $\bmu_1-\bmu_0$ to be sparse, which are more explainable and verifiable in practice.

\label{rem:4.1}
\end{remark}

\setcounter{assume}{0}
\renewcommand{\theassume}{ML\arabic{assume}}

\subsection{Machine learning (nonparametric) setting}\label{sec:thm:ml}

Define that $\|f(\cdot)\|_{Q,q}=:\|f(\U)\|_{Q,q}=:\{\int |f(u)|^qdQ(u)\}^{1/q}$ for any real number $q>0$, function $f(\cdot)$, random variables $\U$ and probability measure $Q$.  Let $P$ denote the probability measure of the observed $\D$. We assume that $K=\Theta(1)$ and introduce the following assumption.
\begin{assume}[Quality of the ML nuisance estimators]
\label{asu:4.2.2}
For each $k\in\{1,2,\ldots,K\}$, 
\begin{align*}
\sup_{\x\in\mathcal{X}}|\widehat r\supmfk(\x)-r_0(\x)|+|\widehat m\supmfk(\x)-m_0(\x)|&=o_{\Pbb}(1);\\
\|\widehat r\supmfk(\cdot)-r_0(\cdot)\|_{P,2}+\|\widehat m\supmfk(\cdot)-m_0(\cdot)\|_{P,2}&=o_{\Pbb}(n^{-1/4}).
\end{align*}
\end{assume}

\begin{remark}
\label{rem:4.2.1}
Similar to Assumptions 3.2 and 3.4 of \cite{chernozhukov2016double}, our Assumption \ref{asu:4.2.2} requires that the ML estimators for $r_0(\cdot)$ and $m_0(\cdot)$ are uniformly consistent and their mean squared errors (MSE) achieve the rate $o_{\Pbb}(n^{-1/4})$. This assumption is also referred as rate doubly robust property \citep{smucler2019unifying} in that it requires production of the MSEs of $\widehat r\supmfk(\cdot)$ and $\widehat m\supmfk(\cdot)$ to be $o_{\Pbb}(n^{-1/2})$. In Appendix \ref{sec:app:fmr}, we provide justification for our proposed FMR procedure to derive that the resulted $\widehat r\supmfk(\cdot)$ satisfies Assumption \ref{asu:4.2.2} as long as the learning algorithm $\Lscr$ satisfies the same strong convergence properties as Assumption \ref{asu:4.2.2}, on all the learning tasks in FMR. Thus, the use of FMR procedure does not actually clip the wings of the ML algorithm being used in our framework.
\end{remark}

We present the asymptotic property of $\widehat\beta\subml$ in Theorem \ref{thm:2} with its proof found in Appendix \ref{sec:app:thm2}.
\begin{theorem}
Denote by $I_0=\ep\bar\psi(\X)Ye^{-\beta_0 A}A\{A-m_0(\X)\}$ and $\sigma_0^2=I_0^{-2}\ep h^2(\D;\beta_0,\eta_0)$. Under Assumptions \ref{asu:1} and \ref{asu:4.2.2}, we have
\[
\sqrt{n}\sigma_0^{-1}(\widehat\beta\subml-\beta_0)=\frac{1}{\sqrt{n}}\sum_{i=1}^n(\sigma_0I_0)^{-1}h(\D_i;\beta_0,\eta_0)+o_{\Pbb}(1),
\]
which weakly converge to ${\rm N}(0,1)$.
\label{thm:2}
\end{theorem}

\begin{remark}
\label{rem:4.2.2}
Since $\bar\psi(\x)={\rm expit}\{-r_0(\x)\}$ and $\widehat\psi\supmfk(\x)={\rm expit}\{-\widehat r\supmfk(\x)\}$ in our ML case, one could show that $\widehat\psi\supmfk(\x)$ achieves the same strong convergence and rate properties as $\widehat r\supmfk(\cdot)$ under Assumption \ref{asu:4.2.2}. While generally speaking, uniform consistency of $\widehat\psi\supmfk(\cdot)$ is sufficient for the desirable conclusion in Theorem \ref{thm:2} so our framework accommodates more flexible choices on $\psi(\x)$, for example, $\psi_{\rm opt}(\x)$ as introduced in Section \ref{sec:method:eff}. We demonstrate this point during the proof in Appendix \ref{sec:app:thm2}.
\end{remark}

\section{Simulation study}\label{sec:sim}

We conduct simulation studies for HD and ML settings separately in Sections \ref{sec:sim:hd} and \ref{sec:sim:ml}, to study the point and interval estimation performance of our method. 

\subsection{High dimensional (parametric) setting}\label{sec:sim:hd}

For the HD parametric setting, we design three data generation configurations introduced as follows to simulate different scenarios of model specification:

\begin{enumerate}

\item[(i)] First, generate $Y$ following $P(Y=1)=1/2$. Then generate $(A,\X\trans)\trans\mid \{Y=j\}\sim \mathcal{N}(\bmu_j,\boldsymbol{\Sigma})$ for $j=0,1$. Specification of $\bmu_j$ and $\boldsymbol{\Sigma}$ are presented in Appendix \ref{sec:app:sim:detail} such that $\beta_0=0.5$, $r_0(\X)=-0.22(X_1+X_2)+0.08(X_3+X_4)$, and $m_0(\X)=-0.1\dot{3}(X_1+X_2+X_3+X_4)$.

\item[(ii)] Generate $Y$ following $P(Y=1)=1/2$ and $(A,\X\trans)\trans\mid \{Y=j\}\sim \mathcal{N}(\bmu_j,\boldsymbol{\Sigma})$ for $j=0,1$. Specification of $\bmu_j$ and $\boldsymbol{\Sigma}_j$ are presented in Appendix \ref{sec:app:sim:detail} such that $\beta_0=0.5$, $r_0(\X)=-0.22(X_1+X_2)+0.08(X_3+X_4)-0.15(X_1X_2+X_1X_3+X_2X_3)$, and $m_0(\X)=-0.1\dot{3}(X_1+X_2+X_3+X_4)$.

\item[(iii)] First generate $\X\sim\mathcal{N}(\bzero,\boldsymbol{\Sigma})$ with $\boldsymbol{\Sigma}\in \mathbb{R}^{p\times p}$ given in Appendix \ref{sec:app:sim:detail}. Then we generate $A$ given $\X$ from a gaussian linear model with unit variance and conditional mean
$$ 
E(A\mid\X)= 0.15X_1+0.15X_2+0.15X_3+0.15X_4+0.075X_1X_2+0.075X_1X_3+0.075X_2X_3.
$$ 
Finally, generate $Y$ by $\Pbb (Y=1\mid A,\X)={\rm expit}(0.5A+0.25X_1+0.25X_2+0.1X_3+0.1X_4)$.
\end{enumerate}
We realize configurations (i)--(iii) with the sample size $n = 1000$, $1500$ or $2000$ separately and the dimension of $\X$ fixed as $p = 200$. Under all these settings, we specify the nuisance models as: $r(\X)=\X\trans\bgamma$ and $m(\X)=\X\trans\balpha$. Then both nuisance models are correctly specified under (i), only $m(\X)$ is correctly specified under (ii), and only $r(\X)$ is correct under (iii). Note that we could not extract the explicit form of $m_0(\X)$ under (iii) since $A$ is generated conditional on $\X$ without fixed $Y=0$. While we still expect the linear model $m(\X)=\X\trans\balpha$ is misspecified under (iii) as there is non-linear terms included in $E(A|\X)$. Implementing details of our HD approach are presented in Appendix \ref{sec:app:num}. Specifically, all the tuning parameters in $\ell_1$-regularized regression are selected using cross-validation among $[0.2(\log p/n)^{1/2},2(\log p/n)^{1/2}]$. We conducted each setting with $300$ repeated simulations.

Table \ref{tab:simu:HD} evaluates the performance of our estimator $\widehat\beta\subhd$ under all settings on its mean square error (MSE), absolute bias and coverage probability (CP) of the 95\% confidence interval (CI) estimated using Gaussian bootstrap multiplier. Under all the settings, our method outputs low root-MSE and bias respectively occupying at most $18\%$ and $7\%$ of the magnitude of the true $\beta_0(=0.5)$ when $n=1000$, and at most $12\%$ and $4\%$ of the $\beta_0$ when $n=2000$. As the sample size $n$ grows, one could see a trend of decaying on the MSEs and bias of our estimator. In addition, under all the settings, our interval estimation has proper CP locating in $\pm0.03$ range of the nominal level $0.95$. Thus, our HD estimator performs steadily well under different model specification scenarios as long as at least one nuisance models are correctly specified.

\begin{table}[!htbp] 
\centering 
\caption{\label{tab:simu:HD} Average mean square error (MSE), average absolute bias (Bias), and average coverage probability (CP) of 95\% CI of our HD estimator with the sample size set as $1000$, $1500$ and $2000$, under configurations (i)--(iii) described in Section \ref{sec:sim:hd}. Number of repetition for each setting is $300$.}
\begin{tabular}{cccccccccc} 
\\[-1.8ex]\hline 
\hline \\[-1.8ex] 
&\multicolumn{3}{c}{Configuration (i)} &\multicolumn{3}{c}{Configuration (ii)}&\multicolumn{3}{c}{Configuration (iii)}  \\ \cmidrule(r){2-4}\cmidrule(r){5-7}\cmidrule(r){8-10}
$n$&$ 1000$ &$1500$&$2000$&$1000$ &$1500$&$2000$& $1000$ &$1500$&$2000$\\
\hline \\[-1.8ex] 
MSE&$0.007$ & $0.006$ & $0.004$ & $0.007$ & $0.006$ & $0.004$ & $0.008$ & $0.004$ & $0.003$ \\ 
Bias&$0.024$ & $0.016$ & $0.017$ & $0.032$ & $0.014$ & $0.020$ & $0.034$ & $0.016$ & $0.012$ \\ 
CP&$0.95$ & $0.94$ & $0.93$ & $0.96$ & $0.92$ & $0.93$ & $0.93$ & $0.95$ & $0.95$ \\ 
\hline
\hline \\[-1.8ex] 
\end{tabular} 
\end{table}

\subsection{Machine learning (nonparametric) setting}\label{sec:sim:ml}

To study our proposed method under the ML setting, we let $\boldsymbol{\Sigma}\in R^{p \times p}$ with $\Sigma_{ii}=1$; $\Sigma_{ij}=0.2$ for $i \neq j$, and generate $\mathcal{N}(\boldsymbol{0},\boldsymbol{\Sigma})$ random vectors and truncate their entries by $(-2,2)$ to obtain $\X$. We then generate $A$ from the gaussian model given $\X$ with unit variance and conditional mean $a_0(\X)= \boldsymbol\zeta_a\trans f_a(\X)$ where $f_a(\x)$ is a non-linear basis function of $\x$ including various types of effects (interaction, indicator and trigonometric function, etc.), defined in Appendix \ref{sec:app:sim:detail} and $\boldsymbol\zeta_a$ represents its loading coefficients also given in  Appendix \ref{sec:app:sim:detail}. Finally, we set $\beta_0=1$, $r_0(\X)=\boldsymbol\zeta_r\trans f_r(\X)$ (see Appendix \ref{sec:app:sim:detail}), and generate $Y$ with $\Pbb (Y=1\mid A,\X)={\rm expit}\{\beta_0 A+ r_0(\X)\}$. We fix $p=20$ and set $n=1000$ or $n=2000$ separately as two settings.

To estimate the nuisance function $r_0(\x)$, we use the FMR procedure with its last step being (\ref{equ:est:r:1}). And the number of fold for cross-fitting is set as $K=5$. For choice of the learning algorithms $\Lscr$, we consider four ML methods and a hybrid method of the ML models introduced as follows.

\begin{enumerate}[(a)]
\item Gradient boosted machines (GBM): an ensemble approach of classification and regression tree (CART) using gradient boosting. Implemented by {\bf R}-package ``gbm" \citep{greenwell2020package}.

\item Random forest (RF): ensemble of CART with bagging. Implemented with {\bf R} package ``RandomForest" \citep{andy2002class}.

\item Support vector machine (SVM): specified with linear kernel and implemented using {\bf R} package ``e1071" \citep{dimitriadou2004r}.

\item Neural network (NN): single hidden layer neural network implemented with {\bf R} package ``nnet" \citep{ripley2016package}.

\item Best nuisance models (Best): similar to \cite{chernozhukov2016double}, we use a simple hybrid method choosing the ML estimator for each nuisance component with the best prediction performance evaluated by cross-validated sum-squared loss.

\end{enumerate}
All the above mentioned ML algorithms are popular in the field of ML and have been considered in the literature of DML \citep{chernozhukov2016double,cui2019bias}. Tuning parameters of the ML models including the number of trees of GBM and RF, the margin of SVM, and the number of units and the weight decay of NN, are selected using the resampling approach of {\bf R} package ``caret" \citep{kuhn2020package}.  We conducted each setting of $n$ again with $300$ repeated simulations.

Table \ref{tab:simu:ml} presents the resulted average MSE, bias and CP of 95\% CI of the estimator $\widehat\beta\subml$ obtained with the five ML modelling strategies for $n=1000$ and $n=2000$ separately. The five approaches have relatively consistent performance in terms of MSE, bias and CP under both settings, with the variation of their MSEs smaller than $0.002$. This demonstrates that performance of our framework is robust to the choice of ML algorithms. While to certain degree, NN has the best performance (with the lowest bias and MSE) when $n=1000$, and GBM and Best have the best performance when $n=2000$. Also, interval estimators of all the approaches achieve proper coverage rates.

\begin{table}[!htbp] 
\centering 
\caption{\label{tab:simu:ml} Average mean square error (MSE), average absolute bias (Bias), and average coverage probability (CP) of 95\% CI of our ML estimator with sample sizes set as $1000$ and $2000$, and the nuisance models estimated by the four ML algorithms and the ``Best" approach described in Section \ref{sec:sim:ml}. Number of repetition for each setting is $300$.}
\begin{tabular}{ccccccccccc} 
\\[-1.8ex]\hline 
\hline \\[-1.8ex] 
&\multicolumn{5}{c}{$n=1000$} &\multicolumn{5}{c}{$n=2000$}\\ \cmidrule(r){2-6} \cmidrule(r){7-11}
& GBM &RF & SVM & NN &Best &GBM &RF & SVM & NN &Best\\
\hline \\[-1.8ex] 
MSE & $0.013$ & $0.013$ & $0.012$ & $0.011$ & $0.013$ & $0.006$ & $0.007$ & $0.007$ & $0.007$ & $0.006$\\
Bias & $0.036$ & $0.048$ & $0.045$ & $0.015$ & $0.037$ & $0.035$ & $0.049$ & $0.046$ & $0.041$ & $0.035$\\
CP & $0.93$ & $0.94$ & $0.92$ & $0.95$ & $0.93$ & $0.94$ & $0.94$ & $0.92$ & $0.94$ & $0.94$\\
\hline
\hline \\[-1.8ex] 
\end{tabular} 
\end{table}

\section{A real example: effect of EC pill on early gestation foetal}\label{sec:real}

In this section, we implement our proposed HD and ML methods to study the effect of emergency contraceptive (EC) pill on the rate of new birth and early gestation foetal death (abortion), by revisiting and exploring the data of a quasi-experimental study based on the policy reform of EC pill's legislation in Chile \citep{bentancor2017assessing}. In their original study, the authors collected all records of birth and foetal death in Chile, as well as a number of municipality level features (education, salary and healthcare, etc) of woman at reproductive age (15--34), in the years around 2008, during which the country was experiencing a reform on the legislation of EC pills. As the consequence of this reform, there is about half of the municipalities in Chile started to provide EC pill freely in 2009, while in the remaining half, EC pill is still not available or restricted in use during that period. This policy was mostly dependent on the political, economic and public health factors characterized by totally $16$ features (denoted as $\boldsymbol{Z}$) such as education spending, public health spending, condom use, and political conservativeness. Thus, the treatment of EC pill ($A=1$ for EC pill accessible; $A=0$ for EC pill not accessible) can be regarded as exogenous for the individuals. 

Let $Y\supone$ denote the indicator for the status of early gestation foetal death in each individual record and $Y\suptwo$ indicate giving new birth (pregnant and did not incur foetal death). Assume that
\begin{align*}
&\Pbb (Y\supone=1\mid A,\Z)={\rm expit}\{\beta_0\supone A+r_0\supone(\Z)\};\\
&\Pbb (Y\suptwo=1\mid A,\Z)={\rm expit}\{\beta_0\suptwo A+r_0\suptwo(\Z)\},
\end{align*}
where $r_0\supone(\cdot)$ and $r_0\suptwo(\cdot)$ are two unknown functions. We are interested in inferring the two parameters $\beta_0\supone$ and $\beta_0\suptwo$ characterizing the log odd ratios (log-OR) of abortion (among the pregnant individuals) and birth (among all individuals) to the treatment of EC pill respectively. To investigate $\beta_0\supone$, we follow a similar strategy as \cite{bentancor2017assessing} that focuses on the individual records at age between 15 and 25, on which early gestation foetal death can be viewed as a reasonable proxy for illegal abortion. Note that the prevalence of $Y\supone$ and $Y\suptwo$ in their corresponding populations are both less than $5\%$, which could cause logistics model unstable to fit. So we randomly downsample the $0$'s in both analysis to make the prevalence of $Y\supone$ and $Y\suptwo$ being $1/4$. This procedure only changes intercepts of the logistic models and does not affect the target parameters. The resulted data set for analyzing $\beta_0\supone$ (abortion) has $n\supone=5,824$ samples. While we only take a subset with $n\suptwo=10,000$ samples for $\beta_0\suptwo$ so that our algorithms will not require excessive computation times.

For our HD approach, we let $\X$ be the $p=175$ dimensional bases joining $\Z$, all the interaction terms of $\Z$ and the three-dimensional natural splines of all the continuous variable in $\Z$. And we specify the nuisance functions as $m(\x)={\rm expit}(\x\trans\balpha^{\scriptscriptstyle (\ell)})$, and $r(\x)=\x\trans\bgamma^{\scriptscriptstyle (\ell)}$ for $j=1,2$. We take $\X=\Z$ as the input covariates of our ML method. The choice and implementation of the ML algorithms are basically the same as in Section \ref{sec:sim:ml}, except that we additionally introduce dropout\footnote{A common and flexible technique in ML research used for regularization and avoiding over-fitting. Here we randomly and independently set each entry of the training covariates matrix as ${\rm N}(0,1)$ variable with probability $0.4$.} for the two tree-based ML algorithms, GBM and RF, to avoid their over-fitting due to that most covariates are of municipality level but the records are of individual level.

Tables \ref{tab:real:death} and \ref{tab:real:birth} present the point estimation, 95\% CI estimation and $p$-values of our HD and ML approaches for $\beta\supone_0$ and $\beta\suptwo_0$, respectively. For $\beta\supone_0$, log-OR of early gestation foetal death to EC pill, point estimations of all methods are positive and around $-0.175\pm0.03$. Their interval estimations are also internally consistent except that SVM outputs a relatively narrow CI and NN includes $0$ near its CI lower bound. And all methods reject the null: $\beta\supone_0=0$ at level $0.05$. This is because NN outputs slightly worse prediction model for $A\mid \X,Y=0$. Note that the result of our hybrid method ``Best" is consistent enough with HD, indicating that our methods under both settings lead to basically the same conclusion. Similar situation occurs to the estimation of $\beta\suptwo_0$ as well. For $\beta\suptwo_0$, all methods reject the null at level $0.05$ and their point and CI estimations are internally consistent (SVM shows a moderate variation from other methods).

Our analysis results reveal that distribution of EC pill could significantly reduce the rate of illegal abortion (in the age group 15--25) and new birth. This is consistent with the results of \cite{bentancor2017assessing} obtained through their municipality level analysis. Although the estimated effect sizes are of different scales\footnote{Their effect is defined in a partially linear model of the abortion/birth rate against the treatment and control variables. While we are measuring the effect of EC pill in a logistic model at the individual level.} and thus incomparable across the two studies, our $p$-values appear to show more significance in that nearly all of them are below $0.05$ while their estimated $p$-values are between $0.05$ and $0.1$. This is because that we use more complex and robust nuisance models to adjust for the confounding effects of $\Z$ and perform our analysis from the individual level granting us to have larger sample sizes.
 
\begin{table}[!htbp] 
\centering 
\caption{\label{tab:real:death} Point estimations, estimated 95\% CI lower/upper bounds (LB/UB) and (two-sided) $p$-values for $\beta\supone_0$ (log odds ratio of early gestation foetal death to the treatment of EC pill) of our HD and ML (with the five different realization described in Section \ref{sec:sim:ml}) approaches.}
\begin{tabular}{ccccccc} 
\\[-1.8ex]\hline 
\hline 
Method & HD & GBM & SVM &RF  & NN &  Best\\
\hline 
\hline
$\hat{\beta_0}$ & -0.175 & -0.207 & -0.144 & -0.181 & -0.154 & -0.181 \\
CI LB & -0.343 & -0.348 & -0.267 & -0.342 & -0.356 & -0.342 \\
CI UB & -0.004 & -0.067 & -0.020 & -0.020 & 0.047 & -0.020 \\
$p$-value & 0.046 & 0.004 & 0.023 & 0.028 & 0.133 & 0.028\\ 
\hline \\[-1.8ex] 
\end{tabular} 
\end{table}

\begin{table}[!htbp] 
\centering 
\caption{\label{tab:real:birth} Point estimations, estimated 95\% CI lower/upper bounds (LB/UB) and (two-sided) $p$-values for $\beta\suptwo_0$ (log odds ratio of new birth to the treatment of EC pill) of our HD and ML (with the five different realization described in Section \ref{sec:sim:ml}) approaches.}
\begin{tabular}{ccccccc} 
\\[-1.8ex]\hline 
\hline 
Method & HD & GBM & SVM &RF  & NN &  Best\\
\hline 
\hline
$\hat{\beta_0}$ & -0.181 & -0.132 & -0.099 & -0.150 & -0.131 & -0.150 \\
CI LB & -0.272 & -0.235 & -0.194 & -0.279 & -0.255 & -0.279 \\
CI UB & -0.091 & -0.030 & -0.005 & -0.021 & -0.008 & -0.021  \\
$p$-value & 0.000 & 0.012 & 0.039 & 0.023 & 0.037 & 0.023\\ 
\hline \\[-1.8ex] 
\end{tabular} 
\end{table}




\section*{Acknowledgements}

The authors thanks their advisor, Tianxi Cai for helpful discussion and comments on this paper.



\bibliographystyle{apalike}
\bibliography{library}

\newtheorem{thmx}{Theorem}
\renewcommand{\thethmx}{\Alph{thmx}}
\setcounter{lemma}{0}
\renewcommand{\thelemma}{A\arabic{lemma}}
\setcounter{equation}{0}
\renewcommand{\theequation}{A\arabic{equation}}
\renewcommand{\thecondition}{A\arabic{condition}}

\clearpage
\newpage
\setcounter{page}{1}
\appendix

\section*{Appendix}

\section{Proof of Theorem \ref{thm:1}}\label{sec:app:thm}

\begin{proof}
By (\ref{equ:dant:r}), we have
\[
n^{-1}\sum_{i=1}^nh(\D_i;\widehat\beta\subhd,\widehat\eta)=n^{-1}\sum_{i=1}^n\widehat\psi(\X_i)\{Y_ie^{-\widehat\beta\subhd A_i}-(1-Y_i)e^{\X_i\trans\widehat\bgamma}\}\{A_i-g(\X_i\trans\widehat\balpha)\}=0.
\]
Our main involvement is to remove the approximation error $n^{-1}\sum_{i=1}^nh(\D_i;\widehat\beta\subhd,\widehat\eta)-h(\D_i;\widehat\beta\subhd,\bar\eta)$, asymptotically. Note that
\begin{align*}
&n^{-1}\sum_{i=1}^nh(\D_i;\widehat\beta\subhd,\widehat\eta)-h(\D_i;\widehat\beta\subhd,\bar\eta)\\
=&n^{-1}\sum_{i=1}^n\widehat\psi(\X_i)(1-Y_i)\{e^{\X_i\trans\bar\bgamma}-e^{\X_i\trans\widehat\bgamma}\}\{A_i-g(\X_i\trans\widehat\balpha)\}\\
&+n^{-1}\sum_{i=1}^n\widehat\psi(\X_i)\{Y_ie^{-\widehat\beta\subhd A_i}-(1-Y_i)e^{\X_i\trans\bar\bgamma}\}\{g(\X_i\trans\bar\balpha)-g(\X_i\trans\widehat\balpha)\}\\
&+n^{-1}\sum_{i=1}^n\{{\rm expit}(-\X_i\trans\widetilde\bgamma)-{\rm expit}(-\X_i\trans\bgamma^*)\}\{Y_ie^{-\widehat\beta\subhd A_i}-(1-Y_i)e^{\X_i\trans\bar\bgamma}\}\{A_i-g(\X_i\trans\bar\balpha)\}\\
=:&\Delta_1+\Delta_2+\Delta_3.
\end{align*}
We handle the terms $\Delta_1$, $\Delta_2$ and $\Delta_3$ separately as follows. First, we have
\begin{align*}
\Delta_1=&n^{-1}\sum_{i=1}^n(1-Y_i)\{\widehat\psi(\X_i)-\bar\psi(\X_i)\}e^{\X_i\trans\bar\bgamma}\{1-e^{\X_i\trans(\widehat\bgamma-\bar\bgamma)}\}\{A_i-g(\X_i\trans\bar\balpha)\}\\
&+n^{-1}\sum_{i=1}^n(1-Y_i)\widehat\psi(\X_i)e^{\X_i\trans\bar\bgamma}\{1-e^{\X_i\trans(\widehat\bgamma-\bar\bgamma)}\}\{g(\X_i\trans\bar\balpha)-g(\X_i\trans\widehat\balpha)\}\\
&+n^{-1}\sum_{i=1}^n(1-Y_i)\bar\psi(\X_i)e^{\X_i\trans\bar\bgamma}\{1-e^{\X_i\trans(\widehat\bgamma-\bar\bgamma)}-\X_i\trans(\widehat\bgamma-\bar\bgamma)\}\{A_i-g(\X_i\trans\widehat\balpha)\}\\
&+n^{-1}\sum_{i=1}^n(1-Y_i)\bar\psi(\X_i)e^{\X_i\trans\bar\bgamma}\{g(\X_i\trans\bar\balpha)-g(\X_i\trans\widehat\balpha)\}\X_i\trans(\widehat\bgamma-\bar\bgamma)\\
&+n^{-1}\sum_{i=1}^n(1-Y_i)\bar\psi(\X_i)e^{\X_i\trans\bar\bgamma}\{A_i-g(\X_i\trans\bar\balpha)\}\X_i\trans(\widehat\bgamma-\bar\bgamma)\\
=:&\Delta_{11}+\Delta_{12}+\Delta_{13}+\Delta_{14}+\Delta_{15}.
\end{align*}
As $\sup_{i\in\{1,\ldots,n\}}|\X_i\trans(\widehat\bgamma-\bar\bgamma)|=O_{\Pbb}(1)$ by Assumption \ref{asu:hd:4}, there exists $M_1=O(1)$ such that with probability approaching $1$,
\begin{equation}
\label{equ:app:a.1}
|1-e^{\X_i\trans(\widehat\bgamma-\bar\bgamma)}|\leq M_1|\X_i\trans(\widehat\bgamma-\bar\bgamma)|,\quad|1-e^{\X_i\trans(\widehat\bgamma-\bar\bgamma)}-\X_i\trans(\widehat\bgamma-\bar\bgamma)|\leq M_1\{\X_i\trans(\widehat\bgamma-\bar\bgamma)\}^2;
\end{equation}
\begin{equation}
\label{equ:app:a.2}
|\widehat\psi(\X_i)-\bar\psi(\X_i)|=\frac{e^{\X_i\trans\bgamma*}|1-e^{\X_i\trans(\widetilde\bgamma-\bgamma*)}|}{(1+e^{\X_i\trans\bgamma*})(1+e^{\X_i\trans\widetilde\bgamma})}\leq M_1|\X_i\trans(\widetilde\bgamma-\bgamma^*)|.
\end{equation}
And by Assumptions \ref{asu:1} and \ref{asu:hd:4}, there exists $M_2=\Theta(1)$ that $\sup_{i\in\{1,\ldots,n\}}|A_i-g(\X_i\trans\bar\balpha)|+|A_i-g(\X_i\trans\widehat\balpha)|\leq M_2$. Consequently, by Assumptions \ref{asu:hd:4} and boundness of $\psi(\cdot)$, we have
\begin{align*}
|\Delta_{11}|\leq& n^{-1}\sum_{i=1}^n M_1^2M_2e^{\X_i\trans\bar\bgamma}|\X_i\trans(\widetilde\bgamma-\bgamma^*)||\X_i\trans(\widehat\bgamma-\bar\bgamma)|\\
\leq&M_1^2M_2\left[n^{-1}\sum_{i=1}^ne^{\X_i\trans\bar\bgamma}\{\X_i\trans(\widetilde\bgamma-\bgamma^*)\}^2\cdot n^{-1}\sum_{i=1}^ne^{\X_i\trans\bar\bgamma}\{\X_i\trans(\widehat\bgamma-\bar\bgamma)\}^2\right]^{1/2}=O_{\Pbb}\left(\frac{s\log p}{n}\right);\\
|\Delta_{12}|\leq&n^{-1}\sum_{i=1}^nM_1e^{\X_i\trans\bar\bgamma}|\X_i\trans(\widehat\bgamma-\bar\bgamma)||g(\X_i\trans\bar\balpha)-g(\X_i\trans\widehat\balpha)|\\
\leq&M_1\left[n^{-1}\sum_{i=1}^ne^{\X_i\trans\bar\bgamma}\{\X_i\trans(\widehat\bgamma-\bar\bgamma)\}^2\cdot n^{-1}\sum_{i=1}^ne^{\X_i\trans\bar\bgamma}\{g(\X_i\trans\bar\balpha)-g(\X_i\trans\widehat\balpha)\}^2\right]^{1/2}=O_{\Pbb}\left(\frac{s\log p}{n}\right);\\
|\Delta_{13}|\leq&n^{-1}\sum_{i=1}^nM_1M_2e^{\X_i\trans\bar\bgamma}\{\X_i\trans(\widehat\bgamma-\bar\bgamma)\}^2=O_{\Pbb}\left(\frac{s\log p}{n}\right);\\
|\Delta_{14}|\leq&n^{-1}\sum_{i=1}^ne^{\X_i\trans\bar\bgamma}|\X_i\trans(\widehat\bgamma-\bar\bgamma)||g(\X_i\trans\bar\balpha)-g(\X_i\trans\widehat\balpha)|=O_{\Pbb}\left(\frac{s\log p}{n}\right),~\mbox{similar to }|\Delta_{12}|.
\end{align*}
By Assumptions \ref{asu:hd:2} and \ref{asu:hd:4},
\[
|\Delta_{15}|\leq \left\|n^{-1}\sum_{i=1}^n(1-Y_i)\bar\psi(\X_i)e^{\X_i\trans\bar\bgamma}\{A_i-g(\X_i\trans\bar\balpha)\}\X_i\right\|_{\infty}\cdot\|\widehat\bgamma-\bar\bgamma\|_1=O_{\Pbb}\left(\frac{s\log p}{n}\right).
\]
Thus, we have $|\Delta_1|=O_{\Pbb}(s\log p/n)$. For $\Delta_2$, we have 
\begin{align*}
\Delta_2=&n^{-1}\sum_{i=1}^n\widehat\psi(\X_i)Y_i(e^{-\widehat\beta\subhd A_i}-e^{-\beta_0A_i})\{g(\X_i\trans\bar\balpha)-g(\X_i\trans\widehat\balpha)\}\\
&+n^{-1}\sum_{i=1}^n\{\widehat\psi(\X_i)-\bar\psi(\X_i)\}\{Y_ie^{-\beta_0 A_i}-(1-Y_i)e^{\X_i\trans\bar\bgamma}\}\{g(\X_i\trans\bar\balpha)-g(\X_i\trans\widehat\balpha)\}\\
&+n^{-1}\sum_{i=1}^n\bar\psi(\X_i)\{Y_ie^{-\beta_0 A_i}-(1-Y_i)e^{\X_i\trans\bar\bgamma}\}\{g(\X_i\trans\bar\balpha)-g(\X_i\trans\widehat\balpha)-g'(\X_i\trans\bar\balpha)\X_i\trans(\widehat\balpha-\bar\balpha)\}\\
&+n^{-1}\sum_{i=1}^n\bar\psi(\X_i)\{Y_ie^{-\beta_0 A_i}-(1-Y_i)e^{\X_i\trans\bar\bgamma}\}g'(\X_i\trans\bar\balpha)\X_i\trans(\widehat\balpha-\bar\balpha)\\
=:&\Delta_{21}+\Delta_{22}+\Delta_{23}+\Delta_{24}.
\end{align*}
Again using Assumptions \ref{asu:1} and \ref{asu:hd:4}, there exists $M_3=\Theta(1)$ such that 
\begin{equation}
\label{equ:app:a.3}
|e^{-\widehat\beta\subhd A_i}-e^{-\beta_0A_i}|\leq M_3|\widehat\beta\subhd-\beta_0|.
\end{equation}
And by Assumption \ref{asu:hd:3} and the mean value theorem, for each $i$, there exits $t_i$ lying between $\X_i\trans\widehat\balpha$ and $\X_i\trans\bar\balpha$ such that
\begin{equation}
\label{equ:app:a.4}
|g(\X_i\trans\bar\balpha)-g(\X_i\trans\widehat\balpha)-g'(\X_i\trans\bar\balpha)\X_i\trans(\widehat\balpha-\bar\balpha)|\leq|g'(\X_i\trans\bar\balpha)-g'(t_i)||\X_i\trans(\widehat\balpha-\bar\balpha)|\leq L\{\X_i\trans(\widehat\balpha-\bar\balpha)\}^2.
\end{equation}
These combined with (\ref{equ:app:a.1}), (\ref{equ:app:a.2}) and Assumptions \ref{asu:1} and \ref{asu:hd:4} lead to that
\begin{align*}
|\Delta_{21}|&=O\left(|\widehat\beta\subhd-\beta_0|\left[n^{-1}\sum_{i=1}^n\{g(\X_i\trans\bar\balpha)-g(\X_i\trans\widehat\balpha)\}^2\right]^{1/2}\right)=O_{\Pbb}\left(\frac{s\log p}{n}\right);\\
|\Delta_{22}|&=O\left(\left[n^{-1}\sum_{i=1}^n\{1+e^{\X_i\bar\bgamma}\}\{\widehat\psi(\X_i)-\bar\psi(\X_i)\}^2\cdot n^{-1}\sum_{i=1}^n\{1+e^{\X_i\bar\bgamma}\}\{g(\X_i\trans\bar\balpha)-g(\X_i\trans\widehat\balpha)\}^2\right]^{1/2}\right)\\
&=O_{\Pbb}\left(\frac{s\log p}{n}\right);\\
|\Delta_{23}|&=O\left(n^{-1}\sum_{i=1}^n\{1+e^{\X_i\bar\bgamma}\}\{\X_i\trans(\widehat\balpha-\bar\balpha)\}^2\right)=O_{\Pbb}\left(\frac{s\log p}{n}\right).
\end{align*}
And by Assumptions \ref{asu:hd:2} and \ref{asu:hd:4}, 
\[
|\Delta_{24}|\leq \left\|n^{-1}\sum_{i=1}^n\bar\psi(\X_i)\{Y_ie^{-\beta_0 A_i}-(1-Y_i)e^{\X_i\trans\bar\bgamma}\}g'(\X_i\trans\bar\balpha)\X_i\right\|_{\infty}\cdot\|\widehat\balpha-\bar\balpha
\|_1=O_{\Pbb}\left(\frac{s\log p}{n}\right).
\]
So we also have $|\Delta_2|=O_{\Pbb}(s\log p/n)$. For $\Delta_3$, we have
\begin{align*}
\Delta_3=&n^{-1}\sum_{i=1}^n\{{\rm expit}(-\X_i\trans\widetilde\bgamma)-{\rm expit}(-\X_i\trans\bgamma^*)\}Y_i(e^{-\widehat\beta\subhd A_i}-e^{-\beta_0 A_i})\{A_i-g(\X_i\trans\bar\balpha)\}\\  
&+n^{-1}\sum_{i=1}^n\Big[\{{\rm expit}(-\X_i\trans\widetilde\bgamma)-{\rm expit}(-\X_i\trans\bgamma^*)-{\rm expit}'(-\X_i\trans\bgamma^*)\X_i\trans(\bgamma^*-\widetilde\bgamma)\}\\
&\quad\quad\quad\quad~\cdot\{Y_ie^{-\beta_0 A_i}-(1-Y_i)e^{\X_i\trans\bar\bgamma}\}\{A_i-g(\X_i\trans\bar\balpha)\}\Big]\\
&+n^{-1}\sum_{i=1}^n{\rm expit}'(-\X_i\trans\bgamma^*)\X_i\trans(\bgamma^*-\widetilde\bgamma)\{Y_ie^{-\beta_0 A_i}-(1-Y_i)e^{\X_i\trans\bar\bgamma}\}\{A_i-g(\X_i\trans\bar\balpha)\}\\
=:&\Delta_{31}+\Delta_{32}+\Delta_{33}.
\end{align*}
Again using the mean value theorem and the fact that $|{\rm expit}'(u)-{\rm expit}'(v)|\leq|u-v|$ for any $u,v\in\mathbb{R}$ (by $|{\rm expit}''(\cdot)|\leq 1$), we have
\[
|{\rm expit}(-\X_i\trans\widetilde\bgamma)-{\rm expit}(-\X_i\trans\bgamma^*)-{\rm expit}'(-\X_i\trans\bgamma^*)\X_i\trans(\bgamma^*-\widetilde\bgamma)|\leq\{\X_i\trans(\bgamma^*-\widetilde\bgamma)\}^2.
\]
Then by (\ref{equ:app:a.2}), (\ref{equ:app:a.3}), and Assumptions \ref{asu:1} and \ref{asu:hd:4}, we have
\begin{align*}
|\Delta_{31}|=&O\left(|\widehat\beta\subhd-\beta_0|\left[n^{-1}\sum_{i=1}^n\{\X_i\trans(\widetilde\bgamma-\bgamma^*)\}^2\right]^{1/2}\right)=O_{\Pbb}\left(\frac{s\log p}{n}\right);\\
|\Delta_{32}|=&O\left(n^{-1}\sum_{i=1}^n\{1+e^{\X_i\bar\bgamma}\}\{\X_i\trans(\widetilde\bgamma-\bgamma^*)\}^2\right)=O_{\Pbb}\left(\frac{s\log p}{n}\right).
\end{align*}
And by Assumptions \ref{asu:hd:2} and \ref{asu:hd:4}, 
\[
|\Delta_{33}|\leq\left\|n^{-1}\sum_{i=1}^n{\rm expit}'(-\X_i\trans\bgamma^*)\{Y_ie^{-\beta_0 A_i}-(1-Y_i)e^{\X_i\trans\bar\bgamma}\}\{A_i-g(\X_i\trans\bar\balpha)\}\X_i\right\|_{\infty}\|\bgamma^*-\widetilde\bgamma\|_1=O_{\Pbb}\left(\frac{s\log p}{n}\right).
\]
Thus, we have $|\Delta_3|=O_{\Pbb}(s\log p/n)$, and by Assumption \ref{asu:hd:5},
\[
n^{-1}\sum_{i=1}^nh(\D_i;\widehat\beta\subhd,\widehat\eta)-h(\D_i;\widehat\beta\subhd,\bar\eta)=O_{\Pbb}\left(\frac{s\log p}{n}\right)=o_{\Pbb}\left(\frac{1}{\sqrt{n}}\right),
\]
which leads to that
\[
n^{-1}\sum_{i=1}^nh(\D_i;\widehat\beta\subhd,\bar\eta)+o_{\Pbb}\left(\frac{1}{\sqrt{n}}\right)=0.
\]
This combined with Remark \ref{rem:2.1} that $\ep h(\D;\beta_0,\bar\eta)=0$ under Assumption \ref{asu:hd:1}, the regularity Assumption \ref{asu:1}, and Theorem 5.21 of \cite{van2000asymptotic}, finally finishes proving Theorem \ref{thm:1}.

\end{proof}

\newpage

\section{Proof of Theorem \ref{thm:2}}\label{sec:app:thm2}

Following the general results of the DML estimator with non-linear Neyman orthogonal score presented in Section 3.3 and Theorem 3.3 of \cite{chernozhukov2016double}, we only need to verify their Assumptions 3.3 and 3.4 on our orthogonal score function $h(\D;\beta,\eta)$, specifically, Conditions \ref{cond:b1} and \ref{cond:b2} presented as follows.

\begin{condition}[Moment condition with Neyman orthogonality]
\label{cond:b1}
It holds that: (a) $\ep h(\D;\beta_0,\eta_0)=0$ and $\Bsc$ contains an interval of length $\Theta(n^{-1/2}\log n)$ centred at $\beta_0$; (b) the map $(\beta,\eta)\rightarrow\ep h(\D;\beta_0,\eta_0)$ is twice continuously Gateaux-differentiable; (c) $|\ep h(\D;\beta,\eta_0)|\geq \min\{|J_0(\beta-\beta_0)|,c_0\}$ where the parameters $\eta_0=\{r_0(\cdot),m_0(\cdot),\bar\psi\}$, $c_0=\Theta(1)$ and $J_0=\partial_{\beta}\ep h(\D;\beta,\eta_0)|_{\beta=\beta_0}=\Theta(1)$; (d) $h(\D;\beta_0,\eta_0)$ obeys Neyman orthogonality, i.e. $\partial_{\eta}\ep h(\D;\beta_0,\eta_0)[\eta-\eta_0]=0$ for all $\eta\in\Esc$ where the parameter space of $\eta$: $\Esc\subseteq \{\eta:\ep |h(\D;\beta_0,\eta_0)[\eta-\eta_0]|<\infty\}$.
\end{condition}

\begin{condition}[Regularity of the score and quality of the nuisance estimators]
\label{cond:b2}
It holds that: (a) $\widehat\eta\supmfk$ belongs to the realization set $\Tsc_n$ for each $k\in\{1,2,\ldots,K\}$, with probability approaching 1 where $\Tsc_n$ satisfies $\eta_0\in\Tsc_n$ and conditions given as follows; (b) The space of $\beta$, $\Bsc$ is bounded and for each $\eta\in\Tsc_n$, the functional space $\Fsc_{\eta}=\{h(\cdot;\beta,\eta):\beta\in\Bsc\}$ is measurable and its uniform covering number satisfies: there exists positive constant $R=\Theta(1)$ and $\nu=\Theta(1)$ such that
\[
\sup_{Q}\log\Nsc(\epsilon\|F_{\eta}(\cdot)\|_{Q,2},\Fsc_{\eta},\|\cdot\|_{Q,2})\leq \nu\log(R/\epsilon),\quad\forall \epsilon\in(0,1],
\]
where $F_{\eta}(\cdot)$ is a measurable envelope function for $\Fsc_{\eta}$: $F_{\eta}(\D)\geq|h(\D;\beta,\eta)|$ for all $\D$ and $\beta\in\Bsc$, and there exists $q>2$ such that $\|F_{\eta}(\cdot)\|_{P,q}=O(1)$; (c) there exists sequence $\tau_n$:
\begin{align*}
&\sup_{\eta=\{r,m,\psi\}\in\Tsc_n,\beta\in\Bsc}|\ep h(\D;\beta,\eta)-\ep h(\D;\beta,\{r_0,m_0,\psi\})|=o(\tau_n),\\
&\sup_{\eta\in\Tsc_n,|\beta-\beta_0|\leq\tau_n}\ep[h(\D;\beta,\eta)-h(\D;\beta_0,\{r_0,m_0,\psi\})]^2+\ep[h(\D;\beta_0,\{r_0,m_0,\psi\})-h(\D;\beta_0,\eta_0)]^2=o(1),\\
&\sup_{r\in(0,1),\eta\in\Tsc_n,|\beta-\beta_0|\leq\tau_n}|\partial_r^2\ep h\{\D;\beta_0+r(\beta-\beta_0),\eta_0+r(\eta-\{r_0(\cdot),m_0(\cdot),\psi\})\}|=o(n^{-1/2});
\end{align*}
and (d) $\ep h^2(\D;\beta_0,\eta_0)=\Theta(1)$.
\end{condition}
We make some simplification and adaptation on the original assumptions in \cite{chernozhukov2016double} to form Conditions \ref{cond:b1}-\ref{cond:b2}, according to our own setting. The only non-trivial change made here is that in Condition \ref{cond:b2} (c), we require
\begin{equation}
\begin{split}
&\sup_{\eta=\{r,m,\psi\}\in\Tsc_n,\beta\in\Bsc}|\ep h(\D;\beta,\eta)-\ep h(\D;\beta,\{r_0,m_0,\psi\})|=o(\tau_n),\\
&\sup_{r\in(0,1),\eta\in\Tsc_n,|\beta-\beta_0|\leq\tau_n}|\partial_r^2\ep h\{\D;\beta_0+r(\beta-\beta_0),\eta_0+r(\eta-\{r_0(\cdot),m_0(\cdot),\psi\})\}|=o(n^{-1/2});
\end{split}    
\label{equ:app:score}
\end{equation}
instead of:
\begin{align*}
&\sup_{\eta\in\Tsc_n,\beta\in\Bsc}|\ep h(\D;\beta,\eta)-\ep h(\D;\beta,\eta_0)|=o(\tau_n),\\
&\sup_{r\in(0,1),\eta\in\Tsc_n,|\beta-\beta_0|\leq\tau_n}|\partial_r^2\ep h\{\D;\beta_0+r(\beta-\beta_0),\eta_0+r(\eta-\eta_0)\}|=o(n^{-1/2}),
\end{align*}
as used in Assumption 3.4 (c) of \cite{chernozhukov2016double}. The first inequality of (\ref{equ:app:score}) is used by \cite{chernozhukov2016double} to derive a preliminary rate for the DML estimator: $|\widehat\beta\subml-\beta_0|=o_{\Pbb}(\tau_n)$ (see their Step 1 of the proof of Lemma 6.3), and the second inequality of (\ref{equ:app:score}) is used in their Step 3 the proof of Lemma 6.3 to process the second order error of
\[
n^{-1}\sum_{k=1}^K\sum_{i\in\Isc_k} h(\D_i;\beta_0,\eta_0)-h(\D_i;\beta,\widehat\eta\supmfk),
\]
uniformly for $\beta$ satisfying $|\beta-\beta_0|\leq\tau_n$.

Note that our modified two assumptions are still sufficient for deriving these results. Since $\ep h(\D;\beta_0,\{r_0,m_0,\psi\})=0$ holds for all $\psi$ (see our Remark \ref{rem:2.1}), there is actually no need to consider $\ep h(\D;\beta,\{r_0,m_0,\psi\})-\ep h(\D;\beta,\{r_0,m_0,\psi_0\})$ when deriving $|\widehat\beta\subml-\beta_0|=o_{\Pbb}(\tau_n)$. While for the Step 3 of \cite{chernozhukov2016double}, one can instead handle the second order error of
\[
n^{-1}\sum_{k=1}^K\sum_{i\in\Isc_k} h(\D_i;\beta_0,\{r_0,m_0,\widehat\psi\supmfk\})-h(\D_i;\beta,\widehat\eta\supmfk),
\]
again using that $\ep h(\D;\beta_0,\{r_0,m_0,\psi\})=0$ holds for all $\psi$, and then remove the remaining error:
\[
n^{-1}\sum_{k=1}^K\sum_{i\in\Isc_k} h(\D_i;\beta_0,\{r_0,m_0,\psi_0\})-h(\D_i;\beta_0,\{r_0,m_0,\widehat\psi\supmfk\}),
\]
through concentration based on the fact $\partial_{\psi} \ep h(\D;\beta_0,\eta_0)[\psi-\psi_0]=0$ and the second inquality of Condition \ref{cond:b2} (c): $\sup_{\eta\in\Tsc_n,|\beta-\beta_0|\leq\tau_n}\ep[h(\D;\beta_0,\{r_0,m_0,\psi\})-h(\D;\beta_0,\eta_0)]^2=o(1)$.

On the other hand, this modification essentially reduces our requirement on the quality of $\widehat\psi\supmfk(\cdot)$, as mentioned in Remark \ref{rem:4.2.2}. As will be seen from the proof below, to fulfill the modified Condition \ref{cond:b2} (c), we do not require $\|\widehat\psi\supmfk(\cdot)-\bar\psi(\cdot)\|_{P,2}=o_{\Pbb}(\tau_n)$ as one shall when following the original version of \cite{chernozhukov2016double} but only needs $\widehat\psi\supmfk(\cdot)$ to be uniformly consistent (though the former one may still be justifiable in our case since we take $\bar\psi(\x)={\rm expit}\{-\bar r(\x)\}$). We now verify Conditions \ref{cond:b1} and \ref{cond:b2} based on our Assumptions \ref{asu:1} and \ref{asu:4.2.2}.

\begin{proof}
Condition \ref{cond:b1} (a) is directly given by our logistic partial model assumption \ref{model}. Condition \ref{cond:b1} (b) is naturally satisfied as $h(\D;\beta,\eta_0+r(\eta-\eta_0))$ is a twice continuously differentiable in $(\beta,r)$. Condition \ref{cond:b1} (c) is directly given by Assumption \ref{asu:1}. And Condition \ref{cond:b1} (d) holds by equation (\ref{equ:derive}), combined with the model assumptions (\ref{model}) and $\ep [A|\X=\x,Y=0]=m_0(\x)$.

By Assumption \ref{asu:4.2.2} and $\bar\psi(\x)={\rm expit}(-\bar r(\x))$, there exists $\zeta_{1,n}=o(1)$ and $\zeta_{2,n}=o(n^{-1/4})$ such that $\widehat\eta\supmfk$ belongs to
\begin{align*}
\Tsc_n=:\Big\{\eta=(r,m,\psi):&\sup_{\x\in\mathcal{X}}|\psi(\x)-\bar\psi(\x)|+|r(\x)-r_0(\x)|+|m(\x)-m_0(\x)|\leq\zeta_{1,n},\\
&\mbox{and}\quad\|r(\cdot)-r_0(\cdot)\|_{P,2}+\|m(\cdot)-m_0(\cdot)\|_{P,2}\leq\zeta_{2,n}\Big\}, 
\end{align*}
with probability approaching $1$, for each $k\in\{1,\ldots,K\}$. We define $\Tsc_n$ of Condition \ref{cond:b2} in this way such that \ref{cond:b2} (a) is satisfied. Now we validate Condition \ref{cond:b2} (b). By Assumption \ref{asu:1} that $A$ and $\beta$ belong to compact sets, and $|m(\x)|\leq m_0(\x)+o(1)$ is uniformly bounded for $\eta=\{r,m,\psi\}\in\Tsc_n$, there exists positive $C_1=\Theta(1)$ such that for $\eta=\{r,m,\psi\}\in\Tsc_n$,
\begin{align*}
h(\D;\beta,\eta)=&\psi(\X)\{Ye^{-\beta A}-(1-Y)e^{r(\X)}\}\{A-m(\X)\}\\
\leq&|\psi(\X)Ye^{-\beta A}\{A-m(\X)\}|+|\psi(\X)e^{r(\X)}(1-Y)\{A-m(\X)\}|\\
\leq&C_1\{\psi(\X)+\psi(\X)e^{r(\X)}\}=C_1\{{\rm expit}(-\psi(\X))+{\rm expit}(\psi(\X))\}\leq C_1+1;\\
\|\partial_{\beta} h(\D;\beta,\eta)\|_{P,2}^2=&\ep[\psi(\X)Ye^{-\beta A}A\{A-m(\X)\}]^2\leq C_1.
\end{align*}
Then by Example 19.7 of \cite{van2000asymptotic}, Condition \ref{cond:b2} (b) holds with $\nu=1$ and $R$ being the diameter of $\Bsc$. Note that Condition \ref{cond:b2} (d) is again directly given by Assumption \ref{asu:1}. It remains to verify Condition \ref{cond:b2} (c). For each $\eta=\{r,m,\psi\}\in\Tsc_n$ and $\beta\in\Bsc$, using the boundness of $\beta$, $A$, $m_0(\x)$ $\psi(\x)$ and $\psi(\x)e^{r(\x)}$, there exists $C_2=\Theta(1)$ such that 
\begin{align*}
&|\ep h(\D;\beta,\eta)-\ep h(\D;\beta,\{r_0,m_0,\psi\})|\\
\leq&|\ep\psi(\X)Ye^{-\beta A}\{m_0(\X)-m(\X)\}|+|\ep\psi(\X)(1-Y)e^{r(\X)}\{m_0(\X)-m(\X)\}|\\
&+|\ep\psi(\X)e^{r(\X)}(1-Y)\{1-e^{r_0(\X)-r(\X)}\}\{A-m(\X_0)\}|\\
\leq&C_2(\|m_0(\X)-m(\X)\|_{P,2}+\|r_0(\X)-r(\X)\|_{P,2}+\|r_0(\X)-r(\X)\|_{P,2}^2)\leq 3C_2\zeta_{2,n}.
\end{align*}
So we take $\tau_n=n^{-1/4}$ and by $\zeta_{2,n}=o(n^{-1/4})$, the first inequality of Condition \ref{cond:b2} (c) is satisfied. Again by the boundness of $\beta$, $A$, $m_0(\x)$ $\psi(\x)$ and $\psi(\x)e^{r(\x)}$; and $\zeta_{1,n},\tau_n=o(1)$, there exists $C_3=\Theta(1)$ such that
\begin{align*}
&\ep[h(\D;\beta,\eta)-h(\D;\beta_0,\{r_0,m_0,\psi\})]^2+\ep[h(\D;\beta_0,\{r_0,m_0,\psi\})-h(\D;\beta_0,\eta_0)]^2\\
\leq&\ep[h(\D;\beta,\eta)-h(\D;\beta_0,\eta)]^2+\ep[h(\D;\beta_0,\eta)-h(\D;\beta_0,\{r_0,m_0,\psi\})]^2\\
&+\ep[h(\D;\beta_0,\{r_0,m_0,\psi\})-h(\D;\beta_0,\eta_0)]^2\\
\leq&\ep[\psi(\X)e^{-\beta_0A}\{e^{(\beta_0-\beta) A}-1\}\{A-m(\X)\}]^2\\
&+\ep[\psi(\X)e^{-\beta_0 A}\{m_0(\X)-m(\X)\}]^2+\ep[\psi(\X)e^{r(\X)}\{m_0(\X)-m(\X)\}]^2\\
&+\ep[\psi(\X)e^{r(\X)}(1-Y)\{1-e^{r_0(\X)-r(\X)}\}\{A-m_0(\X)\}]^2\\
&+\ep[|\bar\psi(\X)-\psi(\X)|\{e^{-\beta_0 A}+e^{r_0(\X)}\}|A-m_0(\X)|]^2\\
\leq&C_3\sup_{a\in\mathcal{A}}|e^{(\beta_0-\beta) a}-1|+C_3\sup_{x\in\mathcal{X}}\{|(1+e^{r(\x)})[\bar\psi(\x)-\psi(\x)]|+|m(\x)-m_0(\x)|+|e^{r(\x)-r_0(\x)}-1|\}\\
\leq&C_3\sup_{a\in\mathcal{A}}|e^{(\beta_0-\beta) a}-1|+C_3\sup_{x\in\mathcal{X}}\{|\bar\psi(\x)-\psi(\x)|+|m(\x)-m_0(\x)|+\{{\rm expit}(r_0(\x))+1\}|e^{r(\x)-r_0(\x)}-1|\}\\
\leq&C_3\{(e^{\tau_nC_3}-1)+2\zeta_{1,n}+2|e^{\zeta_{1,n}}-1|\}=o(1),
\end{align*}
which validates the second inequality of Condition \ref{cond:b2} (c). At last, for each $r\in(0,1)$, denote by $\beta^*=\beta_0+r(\beta-\beta_0)$, $\eta^*=\{r^*,m^*,\psi\}=\{r_0(\cdot),m_0(\cdot),\psi\}+r(\eta-\{r_0(\cdot),m_0(\cdot),\psi\})$. Similar as above deduction, we have that there exists $C_4=\Theta(1)$,
\begin{align*}
&\partial_r^2\ep h\{\D;\beta_0+r(\beta-\beta_0),\eta_0+r(\eta-\{r_0(\cdot),m_0(\cdot),\psi\})\}\\
=&\ep\psi(\X)Ye^{-\beta^*A}A(\beta-\beta_0)\{m_0(\X)-m(\X)\}\\
&+\ep\psi(\X)(1-Y)e^{r^*(\X)}\{r_0(\X)-r(\X)\}\{m_0(\X)-m(\X)\}\\
\leq&C_4|\beta-\beta_0|\cdot\ep|m_0(\X)-m(\X)|+C_4\ep|r_0(\X)-r(\X)|\cdot\ep|m_0(\X)-m(\X)|\\
=&O(\|m(\cdot)-m_0(\cdot)\|_{P,2}^2)+O\{(\beta-\beta_0)^2\}+O(\|r(\cdot)-r_0(\cdot)\|_{P,2}^2)=O(\zeta_{2,n}^2)+o(\tau_n^2)=o(n^{-1/4}).
\end{align*}
\end{proof}
Using the verified Conditions \ref{cond:b1} and \ref{cond:b2}, one can follow nearly the same proof procedures as those of Theorem 3.3 and Lemma 6.3 in \cite{chernozhukov2016double} to prove our Theorem \ref{thm:2}. The only minor difference concerning the processing of $\psi\supmfk$ has been presented as above. As we point out, one can handle this smoothly by first considering $n^{-1}\sum_{k=1}^K\sum_{i\in\Isc_k} h(\D_i;\beta_0,\{r_0,m_0,\widehat\psi\supmfk\})$ when deriving the initial rate and asymptotic expansion of $\widehat\beta\subml$ as $\ep [h(\D_i;\beta_0,\{r_0,m_0,\widehat\psi\supmfk\})|\widehat\psi\supmfk]=0$, and finally concentrate $n^{-1}\sum_{k=1}^K\sum_{i\in\Isc_k} h(\D_i;\beta_0,\{r_0,m_0,\psi_0\})-h(\D_i;\beta_0,\{r_0,m_0,\widehat\psi\supmfk\})$ using that $\partial_{\psi} \ep h(\D;\beta_0,\eta_0)[\psi-\psi_0]=0$ and $\widehat\psi\supmfk(\cdot)$ is uniformly consistent.

\newpage

\section{Justification of the FMR procedure}\label{sec:app:fmr}

In this section, we derive error rates for the ML estimator $\widehat r\supmfk(\cdot)$ resulted from the FMR procedure introduced in Section \ref{sec:method:ml}. Assume that the learning algorithm $\Lscr$ attains the same strong convergence and rate properties as those for $\widehat m\supmfk(\cdot)$ in Assumption \ref{asu:4.2.2}, i.e. for each $j\in\{1,2,\ldots,K\}$ and $k\in\{1,2,\ldots,K\}$:
\begin{align*}
\sup_{\x\in\mathcal{X},a\in\mathcal{A}}|\widehat M\supmfkj(a,\x)-M_0(a,\x)|=&o_{\Pbb}(1);\quad\|\widehat M\supmfkj(\cdot)-M_0(\cdot)\|_{P,2}=o_{\Pbb}(n^{-1/4});\\
\sup_{\x\in\mathcal{X}}|\widehat a\supmfkj(\x)-a_0(\x)|=&o_{\Pbb}(1);\quad\|\widehat a\supmfkj(\cdot)-a_0(\cdot)\|_{P,2}=o_{\Pbb}(n^{-1/4});\\
\sup_{\x\in\mathcal{X}}|\widehat t\supmfk(\x)-t\supmfkj_{\dag}(\x)|=&o_{\Pbb}(1);\quad\|\widehat t\supmfk(\cdot)-t\supmfkj_{\dag}(\cdot)\|_{P,2}=o_{\Pbb}(n^{-1/4}),
\end{align*}
where $t\supmfkj_{\dag}(\x)=:\ep \left[{\rm logit}\{\widehat M\supmfkj(a,\x) \}\Big|\X=\x,\widehat M\supmfkj(\cdot)\right]$. We justify as follows that 
\[
\sup_{\x\in\mathcal{X}}|\widehat r\supmfk(\x)-r_0(\x)|=o_{\Pbb}(1);\quad\|\widehat r\supmfk(\cdot)-r_0(\cdot)\|_{P,2}=o_{\Pbb}(n^{-1/4}).
\]
First, since ${\rm logit}$ is a smooth function, it is not hard to show that
\[
\sup_{\x\in\mathcal{X},a\in\mathcal{A}}|{\rm logit}\{\widehat M\supmfkj(a,\x)\}-{\rm logit}\{M_0(a,\x)\}|=o_{\Pbb}(1);\quad\|{\rm logit}\{\widehat M\supmfkj(\cdot)\}-{\rm logit}\{M_0(\cdot)\}\|_{P,2}=o_{\Pbb}(n^{-1/4}),
\]
under some mild regularity conditions. Then derive the error rate of $\breve\beta\supmfk$ as follows. 
\begin{align*}
&|\Isc_{\text{-}k}|^{-1}\sum_{j=1}^K\sum_{i\in\Isc_{\text{-}k,j}}{\rm logit}\{\widehat M\supmfkj(A_i,\X_i)\}\{A_i-\widehat a\supmfkj(\X_i)\}\\
=&|\Isc_{\text{-}k}|^{-1}\sum_{j=1}^K\sum_{i\in\Isc_{\text{-}k,j}}{\rm logit}\{M_0(A_i,\X_i)\}\{A_i-a_0(\X_i)\}\\
&+|\Isc_{\text{-}k}|^{-1}\sum_{j=1}^K\sum_{i\in\Isc_{\text{-}k,j}}[{\rm logit}\{\widehat M\supmfkj(A_i,\X_i)\}-{\rm logit}\{M_0(A_i,\X_i)\}]\{A_i-a_0(\X_i)\}\\
&+|\Isc_{\text{-}k}|^{-1}\sum_{j=1}^K\sum_{i\in\Isc_{\text{-}k,j}}{\rm logit}\{M_0(A_i,\X_i)\}\{a_0(\X_i)-\widehat a\supmfkj(\X_i)\}\\
&+|\Isc_{\text{-}k}|^{-1}\sum_{j=1}^K\sum_{i\in\Isc_{\text{-}k,j}}[{\rm logit}\{\widehat M\supmfkj(A_i,\X_i)\}-{\rm logit}\{M_0(A_i,\X_i)\}]\{a_0(\X_i)-\widehat a\supmfkj(\X_i)\}\\
=&|\Isc_{\text{-}k}|^{-1}\sum_{j=1}^K\sum_{i\in\Isc_{\text{-}k,j}}{\rm logit}\{M_0(A_i,\X_i)\}\{A_i-a_0(\X_i)\}\\
&+\|\widehat a\supmfkj(\cdot)-a_0(\cdot)\|_{P,2}+\|{\rm logit}\{\widehat M\supmfkj(\cdot)\}-{\rm logit}\{M_0(\cdot)\}\|_{P,2}+O_{\Pbb}(n^{-1/2})\\
=&|\Isc_{\text{-}k}|^{-1}\sum_{j=1}^K\sum_{i\in\Isc_{\text{-}k,j}}{\rm logit}\{M_0(A_i,\X_i)\}\{A_i-a_0(\X_i)\}+o_{\Pbb}(n^{-1/4})+O_{\Pbb}(n^{-1/2}),
\end{align*}
under some mild regularity conditions. Similarly, we have
\begin{align*}
|\Isc_{\text{-}k}|^{-1}\sum_{j=1}^K\sum_{i\in\Isc_{\text{-}k,j}}\{A_i-\widehat a\supmfkj(\X_i)\}^2=|\Isc_{\text{-}k}|^{-1}\sum_{j=1}^K\sum_{i\in\Isc_{\text{-}k,j}}\{A_i-a_0(\X_i)\}^2+o_{\Pbb}(n^{-1/4})+O_{\Pbb}(n^{-1/2}).
\end{align*}
And consequently, by Proposition \ref{prop:1} and 
\[
\breve\beta\supmfk=\frac{\sum_{i\in\Isc_{\text{-}k}}{\rm logit}\{M_0(A_i,\X_i)\}\{A_i-a_0(\X_i)\}}{\sum_{i\in\Isc_{\text{-}k}}\{A_i-a_0(\X_i)\}^2}+o_{\Pbb}(n^{-1/4})=\beta_0+o_{\Pbb}(n^{-1/4}).
\]
Then by Assumption \ref{asu:1} that $\beta_0$ and $|a_0(\x)|$ are bounded and recall $\widehat a\supmfk(\x)=K^{-1}\sum_{j=1}^K\widehat a\supmfkj(\x)$, the estimator $\widehat r\supmfk(\cdot)$ given by equation (\ref{equ:est:r:1}) satisfies that:
\begin{align*}
&\sup_{\x\in\mathcal{X}}|\widehat r\supmfk(\x)-r_0(\x)|\\
\leq&\sup_{\x\in\mathcal{X}}|\widehat t\supmfk(\x)-t_0(\x)|+|\breve\beta\supmfk-\beta_0||a_0(\x)|+|\breve\beta\supmfk||\widehat a\supmfk(\x)-a_0(\x)|\\
\leq&\sup_{\x\in\mathcal{X},j}|\widehat t\supmfk(\x)-t\supmfkj_{\dag}(\x)|+|t\supmfkj_{\dag}(\x)-t_0(\x)|+o_{\Pbb}(1)=o_{\Pbb}(1),
\end{align*}
where $\sup_{\x\in\mathcal{X},j}|t\supmfkj_{\dag}(\x)-t_0(\x)|=o_{\Pbb}(1)$ is a consequence of $\sup_{\x\in\mathcal{X},a\in\mathcal{A}}|{\rm logit}\{\widehat M\supmfkj(a,\x)\}-{\rm logit}\{M_0(a,\x)\}|=o_{\Pbb}(1)$. And
\begin{align*}
&\|\widehat r\supmfk(\cdot)-r_0(\cdot)\|_{P,2}\\
\leq& \|\widehat t\supmfk(\cdot)-t_0(\cdot)\|_{P,2}+|\breve\beta\supmfk-\beta_0|\|a_0(\x)\|_{P,2}+|\breve\beta\supmfk|\|\widehat a\supmfk(\x)-a_0(\x)\|_{P,2}\\
=&\max_{j\in\{1,\ldots,K\}}\|t\supmfkj_{\dag}(\cdot)-t_0(\cdot)\|_{P,2}+o_{\Pbb}(n^{-1/4})\\
=&\max_{j\in\{1,\ldots,K\}}\|\ep[{\rm logit}\{\widehat M\supmfkj(A,\X)\}|\X=\x]-\ep[{\rm logit}\{M_0(A,\X)\}|\X=\x]\|_{P,2}+o_{\Pbb}(n^{-1/4})\\
\leq&\max_{j\in\{1,\ldots,K\}}\|{\rm logit}\{\widehat M\supmfkj(\cdot)\}-{\rm logit}\{M_0(\cdot)\}\|_{P,2}+o_{\Pbb}(n^{-1/4})=o_{\Pbb}(n^{-1/4}).
\end{align*}
Thus, $\widehat r\supmfk(\x)$ satisfies Assumption \ref{asu:4.2.2} under our assumption for the learning tasks of $\Lscr$.

\newpage
\section{Numerical implementation of the HD approach}\label{sec:app:num}

We present and demonstrate the implementation procedure of our HD approach mentioned in Remark \ref{rem:3.1} that uses LASSO instead of the dantzig equation and modifies the construction procedures to make it solvable using the {\bf R}-package \emph{RCAL}. Let $G(u)=\int g(u)du$, and recall that $\widetilde\bgamma$ is some initial estimator obtained through $\ell_1$-regularized logistic regression for $Y$ versus $\{A,\X\}$ and $\widehat\psi(\x)={\rm expit}(-\x\trans\widetilde\bgamma)$. Then we fit 
\begin{equation}
\min_{\balpha\in\mathbb{R}^p}n^{-1}\sum_{i=1}^n(1-Y_i)\widehat\psi(\X_i)e^{\X_i\trans\widetilde\bgamma}\left\{-A_i\X_i\trans\balpha+G(\X_i\trans\balpha)\right\}+\lambda_{\alpha}\|\balpha\|_1,
\label{equ:lasso:m}
\end{equation}
to obtain $\widehat\balpha$. It is not hard to see that the KKT (or subgradient) condition of (\ref{equ:lasso:m}) is equivalent to the $\ell_{\infty}$-constraint in (\ref{equ:dant:m}). And when the link function of $g(\cdot)$ is identity (liner model) or ${\rm expit}(\cdot)$ (logistic model), (\ref{equ:lasso:m}) can be solved using the {\bf R}-package \emph{glmnet} with proper specification of the sample weights, i.e. $(1-Y_i)\widehat\psi(\X_i)e^{\X_i\trans\widetilde\bgamma}$. Then we solve 
\[
n^{-1}\sum_{i=1}^n\widehat\psi(\X_i)\{Y_ie^{-\beta A_i}-(1-Y_i)e^{\X_i\trans\widetilde\bgamma}\}\{A_i-g(\X_i\trans\widehat\balpha)\}=0,
\]
to obtain a preliminary estimator $\widetilde\beta$. It can be shown that when either $r(\x)$ or $m(\x)$ is correctly specified, the estimator $\widetilde\beta$ should approach $\beta_0$ at the rate $O_{\Pbb}\{(s\log p/n)^{1/2}\}$, i.e. the $\ell_2$ errors of $\widetilde\bgamma$ and $\widehat\balpha$. So $\widetilde\beta$ provides a good enough approximation of $\beta_0$ that can be used for the $\ell_1$-regularized (weighted) calibration regression \citep{tan2020regularized} to estimate $\bgamma$:
\begin{equation}
{\rm min}_{\beta\in\mathbb{R},\bfeta\in\mathbb{R}^p}n^{-1}\sum_{i=1}^n\widehat\psi(\X_i)e^{\X_i\trans\widetilde\bgamma}g'(\X_i\trans\widehat\balpha)\{Y_ie^{-\widetilde\beta A_i-\X_i\trans\bgamma}+(1-Y_i)(\widetilde\beta A_i+\X_i\trans\bgamma)\}+\lambda_{\gamma}\|\bgamma\|_1.
\label{equ:lasso:r}    
\end{equation}
Again, KKT condition of (\ref{equ:lasso:r}) corresponds to the $\ell_{\infty}$-constraints in (\ref{equ:dant:r}), though they are not always imposing the same moment conditions: when  the nuisance model $r(\x)$ is misspecified, $\widetilde\bgamma$ and $\widehat\bgamma$ typically have different limits. We use {\bf R}-package \emph{RCAL} to solve (\ref{equ:lasso:r}) with the response taken as $Y_i$, regressors as $\X_i$, sample weight $\widehat\psi(\X_i)e^{\X_i\trans\widetilde\bgamma}g'(\X_i\trans\widehat\balpha)$ and offset $\widetilde\beta A_i$ for each $i$. Denoting the solution of (\ref{equ:lasso:r}) as $\widehat\bgamma$, we finally obtain the estimator $\widehat\beta\subhd$ by solving
\begin{equation}
n^{-1}\sum_{i=1}^n\widehat\psi(\X_i)e^{\X_i\trans(\widetilde\bgamma-\widehat\bgamma)}\{Y_ie^{-\beta A_i}-(1-Y_i)e^{\X_i\trans\widehat\bgamma}\}\{A_i-g(\X_i\trans\widehat\balpha)\}=0.
\label{equ:lasso:beta}   
\end{equation}
Here the final estimating equation is asymptotically equivalent to the second row of (\ref{equ:dant:r}) only when $r(\x)$ is correctly specified ($\widetilde\bgamma$ and $\widehat\bgamma$ have the same limiting values). When $r(\x)$ is misspecified, the orthogonal score function used in (\ref{equ:lasso:beta}), denoted by $h'(\D,\beta_0,\eta)$, is not the same as the $h(\D,\beta_0,\eta)$ used in the main text. We shall point out that this does not hurt the Neyman orthogonality of $h'(\D,\beta_0,\eta)$. It is because that when $r(\x)$ is misspecified but $m(\x)$ is correct (by our model assumption, at least one nuisance model need to be correct), $\partial_{r} h'(\D;\beta_0,\bar\eta)[r-\bar r]$ is naturally satisfied due to the correctness of $m(\x)$ and $\partial_{m} h'(\D;\beta_0,\bar\eta)[m-\bar m]$ is satisfied according to the KKT (moment) condition of (\ref{equ:lasso:r}). When $m(\x)$ is misspecified but $r(\x)$ is correct, $\partial_{m} h'(\D;\beta_0,\bar\eta)[m-\bar m]$ is naturally satisfied and (\ref{equ:lasso:beta}) is asymptotically equivalent with 
\[
n^{-1}\sum_{i=1}^n\widehat\psi(\X_i)\{Y_ie^{-\beta A_i}-(1-Y_i)e^{\X_i\trans\widehat\bgamma}\}\{A_i-g(\X_i\trans\widehat\balpha)\}=0,
\]
as $\widetilde\bgamma$ and $\widehat\bgamma$ approach the true $\bgamma_0$ and the second order errors is asymptotically negligible. So $\partial_{r} h'(\D;\beta_0,\bar\eta)[r-\bar r]$ is satisfied by (\ref{equ:lasso:r}). Thus, our modified construction procedure does not break the theoretical guarantee of $\widehat\beta\subhd$.

\newpage

\section{Additional details of numerical experiments}\label{sec:app:sim:detail}

First, we present the mean vector and covariance matrix used to generate $A$ and $\X$ in Section \ref{sec:sim:hd}:

\begin{enumerate}[(i)]
\item Take $\bmu_1=(0.4,-0.25,-0.25,0,\ldots,0)$, $\bmu_0=\bold{0}$, and
$$({\boldsymbol{\Sigma}}^{-1})_{ij}=
\begin{cases}
1.5& {i=j=1}\\
1.2& {i=j\geq2}\\
0.2& {i=1, 2\leq j \leq5 \ or\  2\leq i \leq5, j=1}\\
0& \text{else}
\end{cases}$$

\item Take $\bmu_1=(0.4,-0.25,-0.25,0,\ldots,0)$, $\bmu_0=\bold{0}$,
    $$({\boldsymbol{\Sigma}_1}^{-1})_{ij}=
    \begin{cases}
1.5& {i=j=1}\\
1.2& {i=j\geq2}\\
0.2& {i=1, 2\leq j \leq5 \ or\  2\leq i \leq5, j=1}\\
0& \text{else}
\end{cases}$$
and 
 $$({\boldsymbol{\Sigma}_0}^{-1})_{ij}=
    \begin{cases}
1.5& {i=j=1}\\
1.2& {i=j\geq2}\\
0.2& {i=1, 2\leq j \leq5 \ or\  2\leq i \leq5, j=1}\\
0.075& {3\leq i \leq4,j=2 \ or\ i=2,3\leq i \leq4 \ or\ i=3,j=4\ or \ i=4,j=3}\\
0& \text{else}
\end{cases}$$

\item Take the covariance of $\X$ as 
 $${\boldsymbol{\Sigma}}_{ij}=
\begin{cases}
0.5& {i=j}\\
0.15& {i\leq4, j \leq4, i\neq j}\\
0& \text{else}
\end{cases}.$$

\end{enumerate}
Then we present the specific choice on the basis functions for data generation in Section \ref{sec:sim:ml}. In specific, we take 
\begin{align*}
f_a(\boldsymbol{x})&=\left\{\frac{1}{1+e^{x_1}},\frac{1}{1+e^{x_2}},{\rm sin}(x_3),{\rm cos}(x_4),I(x_5>0),I(x_6>0),x_7x_8,x_9x_{10}\right\}\trans;\\
\boldsymbol{\zeta}_a&=(2,-2,1,1,0.5,-0.5,0.2,0.2)\trans,
\end{align*}
for $a_0(\x)$, the conditional mean of $A$ given $\X=\x$. And that
\begin{align*}
f_r(\boldsymbol{x})&=\left\{x_1 x_2 x_3, x_4 x_5, x_6^3,{\rm sin}^2(x_7), {\rm cos} (x_8), \frac{1}{1+x_9^2}, \frac{1}{1+e^{x_{10}}}, I(x_{11}>0),I(x_{12}>0)\right\}\trans;\\
\boldsymbol{\zeta}_r &= (0.1,0.1,0.1,-0.5,0.5,1,-1,0.25,-0.25)\trans,
\end{align*}
to specify $r_0(\x)$. 

\end{document}